\documentclass[letter,11pt]{article}
\usepackage[margin=1in]{geometry}
\usepackage[utf8]{inputenc}

\usepackage{times}
\usepackage[T1]{fontenc}    % use 8-bit T1 fonts
\usepackage{hyperref}       % hyperlinks
\usepackage{url}            % simple URL typesetting
\usepackage{booktabs}       % professional-quality tables
\usepackage{amsfonts}       % blackboard math symbols
\usepackage{nicefrac}       % compact symbols for 1/2, etc.
\usepackage{microtype}      % microtypography
\usepackage{graphicx}
\usepackage{amsmath}
\usepackage{amssymb}
\usepackage{amsthm}
\usepackage{csquotes}
\usepackage{thm-restate}
\usepackage{cleveref}
\usepackage{paralist}
\usepackage{xspace}
\usepackage[vlined,linesnumbered,ruled]{algorithm2e}
 
\numberwithin{equation}{section}
\usepackage{enumerate}
\usepackage{listings}
\usepackage{tikz}
\usepackage{csquotes}
\usepackage{cite}
\usepackage{subfigure}
\usepackage{thmtools,thm-restate} % Reusing lemmas and theorems in the appendix
\usepackage{paralist} % compact itemize environments
\usepackage{pgfplots}
\pgfplotsset{compat=newest}
\usepackage{thm-restate}

\usepackage{graphics,graphicx,xcolor,color,url}

	\usepackage[textsize=scriptsize]{todonotes}
	
	\newcommand{\KU}{{\sc KnapExp}\xspace}
	\newcommand{\offKU}{{\sc OfflineKnapExp}\xspace}

	\newcommand{\items}{\ensuremath{\mathcal{I}}\xspace}

	\newcommand{\density}{d}
	\newcommand{\odensity}{\bar{d}}

    \theoremstyle{definition}
    \newtheorem{theorem}{Theorem}

    \newtheorem{remark}{Remark}
    \newtheorem{claim}{Claim}
    \newtheorem{definition}{Definition}

    \newtheorem{property}{Property}

\begin{document}

\date{}
\title{On the Complexity of Knapsack under Explorable Uncertainty: Hardness and Algorithms\thanks{An extended abstract of this paper will appear in the proceedings of \emph{The European Symposium on Algorithms (ESA 2025)}}}

\author{%
Jens Schl{\"o}ter\thanks{Centrum Wiskunde \& Informatica (CWI), Amsterdam, The Netherlands, \texttt{Jens.Schloter@cwi.nl}}
}

\maketitle

\begin{abstract}
    In the \emph{knapsack problem under explorable uncertainty}, we are given a knapsack instance with uncertain item profits. Instead of having access to the precise profits, we are only given \emph{uncertainty intervals} that are guaranteed to contain the corresponding profits. The actual item profit can be obtained via a \emph{query}. The goal of the problem is to adaptively query item profits until the revealed information suffices to compute an optimal (or approximate) solution to the underlying knapsack instance. Since queries are costly, the objective is to minimize the number of queries.
 
    In the offline variant of this problem, we assume knowledge of the precise profits and the task is to compute a query set of minimum cardinality that a third party without access to the profits could use to identify an optimal (or approximate) knapsack solution. We show that this offline variant is complete for the second-level of the polynomial hierarchy, i.e., $\Sigma_2^p$-complete, and cannot be approximated within a non-trivial factor unless $\Sigma_2^p = \Delta_2^p$. Motivated by these strong hardness results, we consider a \enquote{resource-augmented} variant of the problem where the requirements on the query set computed by an algorithm are less strict than the requirements on the optimal solution we compare against. More precisely, a query set computed by the algorithm must reveal sufficient information to identify an approximate knapsack solution, while the optimal query set we compare against has to reveal sufficient information to identify an optimal solution. We show that this resource-augmented setting allows interesting non-trivial algorithmic results.
\end{abstract}

\section{Introduction}
 
The field of \emph{explorable uncertainty} considers optimization problems with uncertainty in the numeric input parameters. Initially, the precise values of the uncertain parameters are unknown. Instead, for each uncertain parameter, we are given an \emph{uncertainty interval} that contains the precise value of that parameter. Each uncertain parameter can be queried to reveal its precise value. The goal is to adaptively query uncertain parameters until we have sufficient information to solve the underlying optimization problem.

 In this paper, we consider \emph{knapsack under explorable uncertainty} (\KU) with uncertain item profits. That is, we are given a set of items $\items$ and a knapsack capacity $B \in \mathbb{N}$. Each item $i \in \items$ has a known weight $w_i \in \mathbb{N}$ and an uncertain profit $p_i \in \mathbb{R}$ that is initially hidden within the known uncertainty interval $I_i$, i.e., $p_i \in I_i$. A \emph{query} of an item $i$ reveals the profit $p_i$. Our goal is to compute a set $P \subseteq \items$ of items with $w(P):= \sum_{i \in P} w_i \le B$ that maximizes the profit $p(P) := \sum_{i \in P} p_i$. We refer to this problem as the \emph{underlying knapsack problem}.
 Since the profits are initially hidden within their uncertainty intervals, we do not always have sufficient information to compute an optimal or even approximate solution for the underlying knapsack problem. Instead, an algorithm for \KU can adaptively query items to reveal their profits until the revealed information suffices to compute an optimal solution for the underlying knapsack instance. As queries are costly, the goal is to minimize the number of queries.

 The \emph{offline version}, sometimes also called \emph{verification problem}, of knapsack under explorable uncertainty (\offKU) assumes full initial access to the profits $p_i$ and asks for a \emph{query set} $Q \subseteq \items$ of minimum cardinality such that access to the profits of the items in $Q$ and access to the uncertainty intervals of the items in $\items \setminus Q$ suffices to compute an optimal solution to the underlying knapsack instance, independent of what the precise profits of the items in $\items \setminus Q$ are. In this work, we mainly focus on studying the offline version of knapsack under explorable uncertainty.
 Most commonly, problems under explorable uncertainy are studied in an \emph{adversarial online setting}, where the uncertain values are unknown, query outcomes are returned in a worst-case manner and algorithms are compared against the optimal solution for the corresponding offline version by using competitive analysis. The complexity of the offline version is a natural barrier for efficiently solving the online version. 

 So far, most problems that have been studied under explorable uncertainty have an underlying problem that belongs to the complexity class P, i.e., can be solved in polynomial time. The seminal work by Kahan~\cite{kahan91queries} on computing the minimum in a set of uncertain values was followed by works on computing the $k$-th smallest uncertain value~\cite{kahan91queries,feder03medianqueries}, computing a minimum spanning tree with uncertain edge weights~\cite{erlebach08steiner_uncertainty,erlebach14mstverification,megow17mst,Erlebach22Learning,MerinoS19,MathwieserC24}, sorting~\cite{halldorsson19sortingqueries,ErlebachLMS23}, shortest path~\cite{feder07pathsqueires}, finding the cheapest set in a given family of sets~\cite{erlebach16cheapestset,MegowS23}, simple geometric problems~\cite{bruce05uncertainty}, stable matchings~\cite{BampisDEMST24}, and other selection problems~\cite{ErlebachLMS23,BampisDEdLMS21}. If we remove the explorable uncertainty aspect, then all of these problems can be solved in polynomial time. 
 
 Even tough these underlying problems are in P, the offline versions of the corresponding problems under explorable uncertainty are often NP-hard. For instance, the offline version of identifying the set of maximal points under explorable uncertainty is NP-hard~\cite{CharalambousH13}, the offline version of the selection problem in~\cite{ErlebachLMS23,BampisDEdLMS21} is NP-hard, and the offline version of the minimum-spanning tree problem under vertex uncertainty is NP-hard~\cite{erlebach14mstverification}. The offline version of selecting the cheapest set is NP-hard~\cite{erlebach16cheapestset} and even hard to approximate within a factor of $o(\log m)$, where $m$ is the number of sets~\cite{MegowS23}. Similarly, the offline version of stable matching under uncertainty is NP-hard to approximate~\cite{BampisDEMST24}. For all of these problems, adding the layer of explorable uncertainty increases the complexity from polynomial-time solvable to NP-hard and leads to interesting algorithmic challenges even tough the underlying problems are easy. However, this observation also raises the following question: 
 
\begin{quote}
    If the underlying problem is already NP-hard, does adding the layer of explorable uncertainty still increase the complexity?     
\end{quote}
 
As a first main result, we answer this question in the affirmative for the offline version of knapsack under explorable uncertainty. More precisely, we show that \offKU is complete for the second level of the polynomial hierarchy, i.e., $\Sigma_2^p$-complete. We even show that, under a certain conjecture ($\Sigma_2^p \not= \Delta_2^p$), no $n^{1-\epsilon}$-approximation is possible for any $\epsilon > 0$, where $n$ is the number of items. The latter can be seen as a natural next step from the inapproximability result given in~\cite{MegowS23}: They show that approximating the offline version of the cheapest set problem is hard to approximate within a factor of $o(\log m)$ by exploiting that it is equivalent to solving a covering integer linear program (ILP) with $m$ constraints, whereas we show our inapproximability result by exploiting that offline \KU can be represented as a covering ILP with an exponential number of constraints. Unfortunately, these extremely strong hardness results pose further challenges:
 
\begin{quote}
    If the hardness of the offline version prevents any non-trivial approximation, is there any hope for interesting algorithmic results in the offline, online or stochastic version?
\end{quote}

Our approach for answering this question is to consider a form of \emph{resource augmentation}. More precisely, we relax the requirements on a solution $Q$ for \offKU: Instead of requiring that querying $Q$ reveals sufficient information to identify an optimal solution for the underlying knapsack problem, we only require sufficient information to identify an $\alpha$-approximate solution. Unfortunately, we can show that, unless P=NP, there is no non-trivial approximation for this relaxed problem variant if we compare against an optimal solution for the relaxed problem variant. However, as a second main result, we show that non-trivial algorithmic results are possible if the requirements on the algorithm's solution are less strict than the requirements on the optimal solution we compare against; we make this more precise in the next section.
  
\subsection{Problem Definition} An \emph{instance} $\mathcal{K}$ of \KU and \offKU is a quintuple $\mathcal{K} = (\items, B, w, p, \mathcal{A})$, where $\items = \{1,\ldots, n\}$ is a set of $n$ items, $B \in \mathbb{N}$ is the knapsack capacity, $w$ is the weight vector with $w_i \in \mathbb{N}_{\le B}$ for all items $i \in \items$, $p$ is the profit vector with $p_i \in \mathbb{R}_{\ge 0}$ for all $i \in \items$, and $\mathcal{A} = \{I_1,\ldots, I_n\}$ is the set of uncertainty intervals such that $p_i \in I_i$ for all $i \in \items$. The quadruple $(\items, B, w, p)$ characterizes the underlying knapsack problem.

As is common in the area of explorable uncertainty, we assume that each uncertainty interval $I_i$ is either open or \emph{trivial}. That is, we either have $I_i = (L_i, U_i)$ for a \emph{lower limit} $L_i$ and an \emph{upper limit} $U_i$, or $I_i =\{p_i\}$. In the latter case, we call both the item $i$ and the uncertainty interval $I_i$ \emph{trivial} and define $U_i=L_i = p_i$. All items that are not trivial are called \emph{non-trivial}. We use $\items_T$ to refer to the set of trivial items. A \emph{query} of an item $i$ reveals the profit $p_i$ and can be seen as replacing the uncertainty interval $I_i=(L_i,U_i)$ with $I_i = \{p_i\}$.

In \offKU, all input parameters are known to the algorithm, while in \KU the profits $p_i$ are initially uncertain. Both problems ask for a \emph{feasible query set} $Q \subseteq \items$ of minimum cardinality. Intuitively as query set $Q \subseteq \items$ is feasible if the revealed information suffices to identify an optimal solution to the underlying knapsack problem \emph{and} to determine the profit of such a solution. We proceed by making this definition more formal.

\paragraph{Packings and Feasible Query Sets.} To formally define feasible query sets, we first define \emph{packings}. A subset of items  $P \subseteq \items$ is a \emph{packing} if $\sum_{i \in P} w_i \le B$. That is, packings are feasible solutions to the underlying knapsack problem. For $P \subseteq \items$ let $p(P) = \sum_{i \in P} p_i$ denote the profit of $P$. We call a packing \emph{optimal} if it maximizes the profit over all packings. We usually use $P^*$ to refer to an optimal packing and $p^* := p(P^*)$ to refer to the \emph{optimal profit}. 

For each packing $P \subseteq \items$ define $U_P := \sum_{i \in P} U_i$, i.e., the term $U_P$ describes an upper limit on the maximum possible profit the packing could potentially have. Note that $U_p$ can be computed even without access to the profits.
By querying items in $P$, the upper limit $U_P$ decreases as we gain more information and can replace the non-trivial uncertainty interval $I_i = (L_i,U_i)$ with $I_i = \{p_i\}$ after we query $i$ and learn the profit $p_i$. 
For $Q \subseteq \items$, we use $U_i(Q)$ to denote the upper limit of $i$ after querying $Q$, i.e., $U_i(Q) = p_i$ if $i \in Q$ and $U_i(Q) = U_i$  otherwise. The upper limit $U_P(Q)$ of packing $P$ after querying a set $Q \subseteq \items$, is 
\[
U_P(Q) := \sum_{i \in P} U_i(Q) = \sum_{i \in P\setminus Q} U_i + \sum_{i \in P \cap Q} p_i = \sum_{i \in P} U_i - \sum_{i \in P \cap Q} (U_i - p_i) = U_P - \sum_{i \in P \cap Q} (U_i - p_i).
\] 

\begin{definition}
    \label{def:feasible}
    A query set $Q \subseteq \items$ is \emph{feasible} if the following two conditions hold:
    \begin{enumerate}
        \item There is a packing $P \subseteq Q \cup \items_T$ with $p(P) = p^*$.
        \item $U_P(Q) \le p^*$ holds for every packing $P \subseteq \items$.
    \end{enumerate}
\end{definition}

The first condition of~\Cref{def:feasible} ensures that querying $Q$ reveals sufficient information to verify that there exists a packing with the optimal profit $p^*$ while the second condition ensures that querying $Q$ reveals sufficient information to verify that no packing can possibly have a larger profit than $p^*$, no matter what the profits of items $i \in \items\setminus Q$ actually are.

Since any packing $P^*$ with $p(P^*) = p^*$ can only satisfy $U_{P^*}(Q) \le p^*$ if $P^* \subseteq Q \cup \items_T$, the second condition of the definition actually implies the first one. In particular, this means that a query set is feasible if and only if it satisfies the constraints of the following ILP. Note that a similar covering point of view was first observed in~\cite{MegowS23} for the cheapest set problem.

\begin{equation}\tag{K-ILP}\label{ILP}
    \begin{array}{lll}
        \min &\sum_{i \in \items} x_i\\
        \text{s.t. }& \sum_{i \in P} x_i \cdot (U_i - p_i) \ge U_P - p^* &\forall P \subseteq \items \colon \sum_{i \in P} w_i \le B \\
        & x_i \in \{0,1\}& \forall i \in \items
    \end{array}
\end{equation}

\paragraph{The offline problem.} In the offline problem \offKU, we are given an instance $\mathcal{K} = (\items, B, w, p, \mathcal{A})$ with full knowledge of all parameters and our goal is to compute a feasible queryset of minimum cardinality, which is equivalent to solving~\eqref{ILP} with full knowledge of all coefficients. We use $Q^*$ to refer to an optimal solution of \offKU.

\paragraph{The online problem.} In the online problem \KU, we are also given an instance $\mathcal{K} = (\items, B, w, p, \mathcal{A})$ but the profits $p_i$ are initially unknown. The goal is to iteratively and adaptively query items $i$ to reveal their profit $p_i$ until the set of queried items is a feasibile query set. This problem can be seen as solving~\eqref{ILP} with uncertain coefficients $U_i-p_i$ and right-hand side values $U_P-p^*$: Querying an item $i$ corresponds to irrevocably setting $x_i = 1$, incurs a cost that cannot be reverted, and reveals the coefficient $(U_i-p_i)$.
% Since the online problem contains uncertainty and cannot be solved to optimality even with unlimited running time, we analyze the online problem using \emph{competitive analysis}. For $\rho \ge 1$, we say that an algorithm is $\rho$-competitive if $|Q(\mathcal{J})| \le \rho \cdot |Q^*(\mathcal{J})|$ for every problem instance $\mathcal{J}$, where $Q(\mathcal{J})$ is the query set computed by the algorithm for instance $\mathcal{J}$ and $Q^*(\mathcal{J})$ is the \emph{offline optimal} solution for instance $\mathcal{J}$. Intuitively, $Q^*(\mathcal{J})$ is the optimal solution to \eqref{ILP} for instance $\mathcal{J}$. The minimum $\rho$ for which the algorithm is $\rho$-competitive is called the competitive ratio of the algorithm.

\paragraph{Relaxations.} As \offKU turns out to admit strong inapproximability results, we introduce the following relaxed notion of feasible query sets. We use $(\alpha,\beta)$-\offKU to refer to the offline problem of computing a $(\alpha,\beta)$-feasible query set of minimum cardinality.

\begin{definition}[($\alpha,\beta$)-feasibility]
    \label{def:alpha:beta:feasibility}
    Let $\alpha, \beta \ge 1$. We say that a query set $Q \subseteq \items$ is $(\alpha,\beta)$-feasible if the following two conditions hold:
    \begin{enumerate}
        \item There is a packing $P$ such that $P \subseteq Q \cup \items_T$ and $p(P) \ge \frac{1}{\alpha} \cdot p^*$.
        \item $U_P(Q) \le \beta \cdot p^*$ for every packing $P$.
    \end{enumerate}
\end{definition}

The first condition of the definition ensures that we can find an $\alpha$-approximation for the underlying knapsack instance by using only queried and trivial items. The second condition ensures that after querying $Q$ no feasible packing can have profit greater than $\beta \cdot p^*$, no matter what the profits of the items in $\items \setminus Q$ are. Thus, querying $Q$ reveals sufficient information to verify that the set $P$ of the first condition is a $\frac{1}{\alpha\beta}$-approximation for the underlying knapsack instance.
We use $Q^*_{\alpha,\beta}$ to refer to a minimum-cardinality $(\alpha,\beta)$-feasible query set. Note that $Q^* = Q^*_{1,1}$, $Q^*$ is $(\alpha,\beta)$-feasible for every $\alpha,\beta \ge 1$ and $|Q^*| = \max_{\alpha \ge 1, \beta \ge 1} |Q^*_{\alpha, \beta}|$.

In contrast to~\Cref{def:feasible}, the first condition of~\Cref{def:alpha:beta:feasibility} does not imply the second one. As a consequence, each $(\alpha,\beta)$-feasible query set is a feasible solution for a variant of~\eqref{ILP} in which we replace $p^*$ with $\beta \cdot p^*$, but the inverse is not necessarily true.

\subsection{Our Results and Outline}

In~\Cref{sec:hardness:of:approximation}, we give several hardness results for \offKU. First, we show that deciding whether $Q = \emptyset$ is an $(\alpha,\beta)$-feasible query set is weakly NP-hard for any $\alpha,\beta \ge 1$, which immediately implies that it is weakly NP-hard to approximate  $(\alpha,\beta)$-\offKU within any bounded multiplicative factor. This hardness result mainly exploits that the optimal solution $p^*$ to the weakly NP-hard~\cite{Karp72} underlying knapsack problem appears on the right-hand sides of~\eqref{ILP}. Then, we move on to show that \offKU is $\Sigma_2^p$-complete, which intuitively means that the problem remains hard even if we are given an oracle for deciding problems in NP. Since such an oracle allows us to compute $p^*$, the reason for the $\Sigma_2^p$-hardness is not the appearance of $p^*$ in the right-hand sides but the exponential number of constraints in~\eqref{ILP}. In fact, we prove this hardness result via reduction from \emph{succinct set cover}\cite{Umans01,Umans99}, which is a set cover variant where the element are only given implicitly, i.e., the number of set cover constraints is exponential in the encoding size of the problem.  Exploiting a result by~\cite{Ta2021}, our reduction also shows that there is no $n_0^{1-\epsilon}$-approximation for any $\epsilon > 0$ unless $\Sigma_2^p = \Delta_2^p$, where $n_0 := |\items \setminus \items_T$|.

In~\Cref{sec:offline:algorithms}, we design algorithms for $(\alpha,\beta)$-\offKU for different values of $\alpha$ and $\beta$. Since the hardness results prevent any non-trivial results when comparing against $|Q^*_{\alpha,\beta}|$, we analyze our algorithm by comparing their solutions to $|Q^*|$ instead. This can be seen as a form of resource augmentation as the requirements on the algorithms's solution are less strict than the requirements on the optimal solution we compare against. To achieve our algorithmic results, we treat the two conditions for $(\alpha,\beta)$-feasible query sets (cf.~\Cref{def:alpha:beta:feasibility}) as separate subproblems: (i) Compute a query set $Q_1$ such that there exists a packing $P \subseteq Q_1 \cup \mathcal{I}_T$ with $p(P) \ge \frac{1}{\alpha} p^*$, and (ii) Compute a query set $Q_2$ such that $U_P(Q_2) \le \beta p^*$ holds for all packings $P$.
First, we show how to solve subproblem (i) in polynomial-time for $\alpha = \frac{1}{1-\epsilon}$ with a set $Q_1$ such that $|Q_1| \le |Q^*|$. Our algorithm for this subproblem exploits existing results for the two-dimensional knapsack problem. 
For (ii), we first show how to solve the problem in pseudopolynomial time for $\beta = 2 + \epsilon$ with the guarantee $|Q_2| \le |Q^*|$. The algorithm is based on solving an \offKU variant that only considers packings that correspond to certain greedy solutions. We justify the pseudopolynomial running-time by showing weak NP-hardness for this \offKU variant. By considering a relaxed version of that problem, we manage to solve problem (ii) for $\beta = 4 + \epsilon$ with $|Q_2| \le |Q^*|$ in polynomial time.
Combining the results for both subproblems yields a pseudopolynomial algorithm that computes a $(\frac{1}{1-\epsilon},2+2\epsilon)$-feasible query set $Q$ and a polynomial time algorithm that computes a $(\frac{1}{1-\epsilon},4+4\epsilon)$-feasible query set $Q$. In both cases, $|Q| \le 2 \cdot |Q^*|$.

\subsection{Further Related Work} 
Meißner~\cite{meissner18querythesis} gives an adversarial lower bound of $n$ for \KU that holds even if the instance only has two different weights, preventing any non-trivial adversarial results. However, Megow and Schlöter~\cite{MegowS23} show that this lower bound does not hold in a stochastic setting where the profits $p_i$ are drawn from their intervals $I_i$ according to an unknown distribution that satisfies $\Pr[p_i \le \frac{U_i+L_i}{2}] \le \tau$ for a threshold parameter $\tau$. However, their result only breaks that particular lower bound instance and does not imply general stochastic results for \KU.
%   Furthermore,~\cite{MegowS23} how to approximate uncertain ILPs of a form similar to~\eqref{ILP} in their stochastic setting. Unfortunately, their approximation factor has a logarithmic dependency on the number constraints, which does not give any non-trivial results for our problem.

Goerigk et al.~\cite{GoerigkGISS15} consider a knapsack problem under uncertainty in a different query setting. In their problem, the profits are known and the weights are uncertain. Furthermore, there is a budget on the queries that an algorithm is allowed to make. These differences lead to a problem fundamentally different from \KU.

Maehara and Yamaguchi~\cite{maehara2020} consider packing ILPs with uncertainty in the cost coefficients. The cost coefficients can be queried. The key difference to the setting of explorable uncertainty is that they are interested in bounding the absolute number of queries instead of comparing against the optimal feasible query set. We remark that this is an important distinction between explorable uncertainty and many other query models. For example, the same distinction applies to a line of research that studies queries that reveal the \emph{existence} of entities instead of numeric values, e.g., the existence of edges in a graph,~c.f.~\cite{Blum2020,GoemansV06,Vondrak07,Behnezhad2022,BehnezhadDH20,DughmiKP23}. For example, Behnezhad et al.~\cite{Behnezhad2022} considered vertex cover in a stochastic setting and showed that it can be approximated within a factor of $(2+\epsilon)$ with only a constant number of queried edges per vertex.

\section{Hardness of Approximation}
\label{sec:hardness:of:approximation}

We start by showing our hardness results for $(\alpha,\beta)$-\offKU.
Not surprisingly, the appearance of $p^*$ in the right-hand sides of~\eqref{ILP} suffices to render $(\alpha,\beta)$-\offKU weakly NP-hard. The following proposition shows that even deciding whether a given set $Q$ is $(\alpha,\beta)$-feasible is already weakly NP-hard. 

%  The proof of the proposition is a fairly straightforward reduction from weakly NP-hard~\cite{Karp72} classical knapsack problem, which we defer to~\Cref{app:hardness:of:approximation}.

\begin{restatable}{proposition}{propHardness}
    \label{prop:hardness}
    Deciding if $Q = \emptyset$ is $(\alpha,\beta)$-feasible is weakly NP-hard for any $\alpha,\beta \ge 1$. 
    %This holds even for $Q = \emptyset$. 
\end{restatable}

\begin{proof}
    We reduce from the decision variant of knapsack. Consider a given knapsack instance and a parameter $D$ with the goal to decide wether $p^* \ge D$ holds for the optimal profit $p^*$.
    
    We construct an instance of the offline problem by using the given knapsack instance with trivial uncertainty intervals, i.e., $I_i = \{p_i\}$ for all items $i$. We add one additional item $n+1$ with uncertainty interval $I_{n+1} = (0, \beta \cdot D)$, value $p_{n+1} = \epsilon$ for a sufficiently small $\epsilon > 0$, and weight $w_{n+1}=B$.
    Furthermore, we construct the query set $Q = \emptyset$.

    If $p^* \ge D$, then $U_{n+1}(Q) = \beta \cdot D \le \beta \cdot p^*$, which implies that $Q$ satisfies the second condition of~\Cref{def:alpha:beta:feasibility}. Since $p_{n+1} = \epsilon$, there also is a packing $P \subseteq Q \cup \items_T = \items \setminus \{n+1\}$ with $p(P) \ge \frac{p^*}{\alpha}$. Thus, $Q = \emptyset$ is $(\alpha,\beta)$-feasible. 

    If $p^* < D$, then $U_{n+1}(Q) = \beta \cdot D > \beta \cdot p^*$. Thus, $Q$ is not $(\alpha,\beta)$-feasible as it does not satisfy the second condition of~\Cref{def:alpha:beta:feasibility}.

    In conclusion, $Q$ is $(\alpha,\beta)$-feasible if and only if $p^* \ge D$.
\end{proof}

This means that distinguishing between instances that can be solved without any query and instances that need at least one query is weakly NP-hard, which implies the following:

\begin{restatable}{corollary}{corHardness}
    \label{cor:hardness}
    It is weakly NP-hard to approximate $(\alpha,\beta)$-\offKU within any bounded multiplicative factor.
\end{restatable}

\begin{proof}
    Assume there exists a multiplicative $\gamma(n)$-approximation for the $(\alpha,\beta)$-offline problem and some bounded function $\gamma$.
    
    Given any knapsack instance, we can use the reduction of~\Cref{prop:hardness} to construct an instance of offline knapsack under explorable uncertainty.
    
    If the $\gamma(n)$-approximation executes queries to solve this instance, then $Q = \emptyset$ must be infeasible as $\frac{|A|}{|Q^*|} = \frac{|A|}{0}$ would be unbounded for the set of items $A$ queried by the $\gamma(n)$-approximation and the optimal solution $Q^* =\emptyset$. If the $\gamma(n)$-approximation does not execute any queries, then the set $Q = \emptyset$ must be feasible as the approximation must compute a feasible query set. 
    
    Thus, the $\gamma(n)$-approximation executes queries if and only if $Q = \emptyset$ is infeasible. Per proof of~\Cref{prop:hardness}, this means that the approximation algorithm executes queries if and only if $p^* < D$. This implies that an $\gamma(n)$-approximation would solve the given decision problem knapsack instance in polynomial time. 
\end{proof}

\Cref{prop:hardness} is also an indicator that $(\alpha,\beta)$-\offKU might not be in NP, as verifying whether a query set is $(\alpha,\beta)$-feasible is already NP-hard. For $\alpha, \beta = 1$, we make this observation more formal by proving that \offKU is complete for the second level of the polynomial hierarchy, i.e., $\Sigma_2^p$-complete. Intuitively, the class $\Sigma_2^p$ contains problems that, given an oracle for deciding problems in NP, can be solved in non-deterministic polynomial time. 
Similarly, the class $\Delta_2^p$ contains problems that, given the same type of oracle, can be solved in deterministic polynomial time. Hardness and completeness for the class $\Sigma_2^p$ are defined in the same way as for the class NP. For a more formal introduction, we refer to~\cite{AroraBarak2009}.
Under the conjecture that $\sum_2^p \not= \text{NP}$, the $\sum_2^p$-completeness implies that \offKU is not in NP, and under the conjecture  $\sum_2^p \not= \Delta_2^p$ it cannot be solved optimally in polynomial time even when given an oracle for deciding problems in NP.

\begin{theorem}
    \label{thm:level2:hardness}
   \offKU is $\Sigma_2^p$-complete.
\end{theorem}

\begin{proof}
   % We separately show membership in $\Sigma_2^p$ and $\Sigma_2^p$-hardness.
    We show $\Sigma_2^p$-membership
   %  Membership in $\Sigma_2^p$ can be shown by observing that an $(\alpha,\beta)$-feasible query set $Q$ can serve as a certificate for the decision variant of the $(\alpha,\beta)$-offline problem, provided that we have access to an oracle deciding problems in NP. We defer the proof
   in~\Cref{app:hardness:of:approximation} and focus here on proving $\Sigma_2^p$-hardness.
   Our proof is via reduction from \emph{succinct set cover}, which is known to be $\sum_2^p$-complete~\cite{Umans01,Umans99}. In the same way as for NP-hardness proofs, we need to give a polynomial time reduction to the decision problem variant of \offKU such that the constructed instance is a YES-instance if and only if the given succinct set cover instance is a YES-instance.

    \paragraph{Succinct set cover.} We are given $n$ decision variables $x_1,\ldots, x_n$, $m$ propositional logic formulas $\phi_1,\ldots,\phi_m$ over these variables and an integer parameter $k$. Each formula $\phi_j$ is in $3$-DNF form\footnote{\emph{Disjunctive normal form (DNF)} refers to a disjunction of conjunctions, i.e., $\phi_j = C_{j,1} \lor \ldots \lor C_{j,k_j}$, where each clause $C_{j,k}$ is a conjunction of literals. In 3-DNF, each $C_{j,k}$ contains exactly three literals. A formula in DNF is satisfied by a variable assignment if at least one clause is satisfied by the assignment.} and we use $S_j$ to denote the set of $0$-$1$-vectors (variable assignments) that satisfy $\phi_j$. The formulas $\phi_j$ have to satisfy $\bigcup_{j \in \{1,\ldots,m\}} S_j = \{0,1\}^n$, i.e., each variable assignment satisfies at least one formula $\phi_j$. The goal is to find a subset $S \subseteq \{1,\ldots,m\}$ such that $\bigcup_{j \in S} S_j = \{0,1\}^n$ and $|S| \le  k$. 
    We assume that each variable occurs as a literal in at least one formula. If not, we can just remove the variable and obtain a smaller instance.
    Succinct set cover can be interpreted as a set cover variant where the elements and sets are only given \emph{implicitly} and not as an explicit part of the input.

    % Intuitively, in succinct set cover, each possible variable assignment in $\{0,1\}^n$ corresponds to a set cover element and each $S_j$ for a formula $\phi_j$ corresponds to a set cover set. In contrast to the \enquote{normal} set cover problem, the elements and sets are only given \emph{implicitly} and are not an explicit part of the input.  

    \paragraph{Main reduction idea.} Before we give the technical details of the reduction, we first sketch the basic idea. In particular, we describe the properties that we want the constructed instance to have. In the technical part, we then describe how to actually achieve these properties. 
    At its core, the reduction will use the knapsack weights to encode structural information of the input instance into the constructed instance. The idea to use numerical weights to encode constraints is quite natural in NP-hardness proofs for weakly NP-hard problems, see e.g.~the NP-hardness proof for the subset sum problem given in~\cite{AlgorithmDesign}. The usage of the knapsack weights in our reduction is on some level inspired by such classical reductions, but requires several new ideas to handle the implicit representation of the input problem.
    
    First, we introduce a single trivial item $i^*$ with $w_{i^*} = p_{i^*} = B$. This item alone fills up the complete knapsack capacity $B$, which we define later, and the instance will be constructed in such a way that $p^* = p_{i^*} = B$ is the maximum profit of any packing. Thus, only packings $P$ with $U_P > p^*$ induce non-trivial constraints in~\eqref{ILP}. We design the instance such that $U_P \ge p^*$ only holds for packings that use the full capacity.

   \begin{property}
       \label{prop:major:1}
       A packing $P$ of the constructed instance satisfies $U_P \ge p^*$ only if $w(P) = B$.
   \end{property}

    Next, we want each packing $P$ with $w(P)=B$ and $U_P > p^*$ to represent a distinct variable assignment in $\{0,1\}^n$. To this end, we introduce a set $X$ of $2n$ items, two items $v_i$ and $\bar{v}_i$ for each variable $x_i$ with $i \in \{1,\ldots,n\}$. Intuitively, $v_i$ represents the positive literal $x_i$ and $\bar{v}_i$ represents the negative literal $\neg x_i$. We say that a subset $X' \subseteq X$ \emph{represents} a variable assignment if $|X' \cap \{v_i,\bar{v}_i\}|=1$ for all $i \in \{1,\ldots,n\}$. We design our instance such that the packings $P$ with $w(P)= B$ and $U_P > p^*$ exactly correspond to the variable assignments in $\{0,1\}^n$. Note that this excludes the packing $P=\{i^*\}$ as this packing has $w(P) = B$.

    \begin{restatable}{property}{propTwo}
       \label{prop:major:2}
       If $w(P) = B$ and $U_P > p^*$, then $P \cap X$ represents a variable assignment. Each variable assignment is represented by at least one $P$ with $w(P) = B$ and $U_P > p^*$.
   \end{restatable}

   If the first two properties hold, then all non-trivial constraints in the ILP~\eqref{ILP} for the constructed instance correspond to a packing $P$ with $w(P)=B$ and $U_P> p^*$ and, thus, to a variable assignment of the given succinct set cover instance. Furthermore, each variable assignment is represented by at least one active constraint. With the next property, we want to ensure that each possible query, i.e., each non-trivial item, corresponds to a succinct set cover formula $\phi_j$. To this end, we introduce the set of items $Y = \{y_1,\ldots,y_m\}$. These items will be the only non-trivial items in the constructed instance, so each possible query is to an element of $Y$. 
   Next, we want to achieve that querying an item $y_j$ suffices to satisfy all constraints of~\eqref{ILP} for packings $P$ with $w(P)=B$ and $U_P> p^*$ that represent a variable assignment which satisfies formula $\phi_j$, and does not impact any other constraints.
   Formally, we would like to design our instance such that the following property holds.
   % Furthermore, we want that a query corresponding to $\phi_j$ covers all constraints for packings $P$ that represent a variable assignment which satisfies formula $\phi_j$. To this end, we introduce the set of items $Y = \{y_1,\ldots,y_m\}$. These items will be the only non-trivial items in the constructed instance, so each possible query is to an element of $Y$. Our goal is to design the instance in such a way that a packing $P$ with $U_P > p^*$ \emph{has to} contain all $y_j$ such that $\phi_j$ is satisfied by $X \cap P$. Furthermore, we want that a query to a single contained $y_j$ is enough to satisfy the constraint of $P$. Formally:
   \begin{property}
       \label{prop:major:3}
       For each packing $P$ with $U_P > p^*$ and each $y_j \in Y$: $y_j \in P$ if and only if $X \cap P$ represents a variable assignment that satisfies $\phi_j$. If $y_j \in P$, then $U_P-p^* \le U_{y_j} - p_{y_j}$.
   \end{property}
   %  \begin{leftbar}
   %      3.~For each packing $P$ with $U_P \ge p^*$ and each $y_j \in Y$: $y_j \in P$ if and only if $X \cap P$ represents a variable assignment that satisfies $\phi_j$. If $y_j \in P$, then $U_P-p^* \le U_{y_j} - p_{y_j}$.
   %  \end{leftbar}    
    If we manage to define our reduction in such a way that the three properties are satisfied, it is not hard to show correctness (see~\Cref{app:hardness:of:approximation} for the second direction):

    \textbf{First Direction.} If there is an index set $S$ with $|S|\le k$ that is a feasible solution to the succinct set cover instance, then each possible variable assignment must satisfy at least one formula $\phi_j$ with $j \in S$. We claim that $Q = \{y_j \mid j \in S\}$ is a feasible query set for the constructed \offKU instance. To this end, consider an arbitrary packing $P$ with $U_P > p^*$, which are the only packings that induce non-trivial constraints in the corresponding~\eqref{ILP}.
   By~\Cref{prop:major:1},  we have $w(P) = B$. \Cref{prop:major:2} implies that $X \cap P$ represents some variable assignment $\varphi$ and \Cref{prop:major:3} implies that $y_j \in P$ for all $\phi_j$ that are satisfied by $\varphi$. By assumption that $S$ is a feasible succinct set cover solution, we get $Q \cap P \not= \emptyset$. \Cref{prop:major:3} implies  $U_P(Q) \le U_P(\{y_j\}) \le p^*$ for $y_j \in Q \cap P$. Thus, $Q$ satisfies the constraint of $P$ in the~\eqref{ILP} for the constructed instance.  
   % As this holds for all $P$ with $U_P > p^*$, the set $Q$ is a feasible solution for the constructed \offKU instance.

   %  \textbf{Second Direction.} The second direction can be shown in a similar way; we defer the formal proof to~\Cref{app:hardness:of:approximation}.

    \paragraph{Technical reduction} It remains to show how to actually construct an \offKU instance that satisfies the three properties.
    Given an instance of succinct set cover, we construct an instance of \offKU consisting of four sets $X, \Phi$, $A$ and $L$ of items such that $\items = X \cup \Phi \cup A \cup L$. The set $X$ is defined exactly as sketched above, the set $\Phi$ contains the set $Y =\{y_j,\ldots,y_m\}$ as introduced above, and $L:= \{i^*\}$ for the item $i^*$ with $w_{i^*} = p_{i^*} = B$. All further items will be used to ensure the three properties. 
        
    Conceptionally, we construct \emph{several} partial weights for each item $i$ that will later be combined into a single weight. For each item $i$, we construct two weights $w_{i,\phi_j}$ and $w_{i,\rho_j}$ for each formula $\phi_j$, and a single weight $w_{i,x}$. Similarly, we break down the knapsack capacity into multiple partial knapsack capacities $B_x$, and $B_{\phi_j},B_{\rho_j}$ for each $\phi_j$. Intuitively, the full weight $w_i$ of an item $i$ will be the concatenation of the decimal representations of the partial weights, i.e., $w_i = w_{i,x}w_{i,\rho_m}\cdots w_{i,\rho_{1}}w_{i,\phi_m}\cdots w_{i,\phi_{1}}$, and $B$ will be the concatenation of the partial capacities. We make this more precise in~\Cref{app:hardness:of:approximation} after defining all partial weights in such a way that the following property holds.
   
    \begin{restatable}{property}{propFour}
       \label{prop:4} 
       For each packing $P$, it holds $\sum_{i \in P} w_i = B$ if and only if $\sum_{i \in P} w_{i,x} = B_x$, and $\sum_{i \in P} w_{i,\phi_j} = B_{\phi_j}$ and  $\sum_{i \in P} w_{i,\rho_j} = B_{\rho_j}$ for all $j \in \{1,\ldots,m\}$.
    \end{restatable}

    For now, we operate under the assumption that~\Cref{prop:4} holds and focus on the partial weights and capacities, and proceed by defining remaining parts of the construction.

    \paragraph{Definition of the $w_x$-weights:} As formulated in~\Cref{prop:major:2}, we would like each packing $P$ with $w(P) = B$ and $U_P> p^*$ to represent a variable assignment. To this end, we need such a packing to contain exactly one item of $\{v_i,\bar{v}_i\}$ for each $i \in \{1,\ldots,n\}$. To enforce this, we use the partial $w_x$-weights and the partial capacity $B_x$. In particular, we define $w_{v_i,x} = w_{\bar{v}_i,x} = 10^i$ for each $i \in  \{1,\ldots,n\}$, and $B_x = \sum_{i=1}^n 10^i$. For all items $j \in \mathcal{I} \setminus X$, we define $w_{j,x} = 0$, which immediately implies the following property:    

   \begin{property}
       \label{prop:5} 
    $P$ satisfies $\sum_{i \in P} w_{i,x} = B_x$ iff $P \cap X$ represents a variable assignment. 
   \end{property}

    \paragraph{Definition of the set $A$ and the $w_{\phi_j}$-weights:} Define $A = \bigcup_{j \in \{1,\ldots,m\}} A_{\phi_j}$ for sets $A_{\phi_j}$ to be defined below. For formula $\phi_j$, let $k_j$ denote the number of clauses in $\phi_j$ and let $C_{j,1},\ldots, C_{j,k_j}$ denote these clauses. For each $C_{j,k}$, we add four items $a_{j,k,0},a_{j,k,1},a_{j,k,2},a_{j,k,3}$ to set $A_{\phi_j}$. The idea is to define the partial $w_{\phi_j}$-weights and the partial $B_{\phi_j}$ capacity in such a way that the following  property holds.
    
    \begin{restatable}{property}{propSix}
       \label{prop:6} 
       For each $\phi_j$, a packing $P$ satisfies $\sum_{i \in P} w_{i,\phi_j} = B_{\phi_j}$ and $\sum_{i \in P} w_{i,x} = B_x$ iff $a_{j,k,0} \in P$ for each clause $C_{j,k}$ that is satisfied by the assignment represented by $X \cap P$.   
   \end{restatable}

    To achieve the property, we first define the partial $w_{\phi_j}$-weights for the items $X$.
    For a literal $x_i$, let $\mathcal{C}_{x_i,j} := \{ k \mid x_i \text{ occurs in } C_{j,k}\}$ and define  $\mathcal{C}_{\neg x_i,j}$ in the same way.
    We define the weights $w_{v_i, \phi_j}$ and $w_{\bar{v}_i, \phi_j}$ as $\sum_{k \in \mathcal{C}_{x_i,j}} 10^{k_j + k - 1}$ and $\sum_{k \in \mathcal{C}_{\neg x_i,j}} 10^{k_j + k - 1}$, respectively.
    Intuitively, digit $k_j+k$ of the sum $\sum_{h \in X \cap P} w_{h,\phi_j}$ of a packing $P$ with $\sum_{i \in P} w_{i,x} = B_x$ indicates whether the assignment represented by $X\cap P$ satisfies clause $C_{j,k}$ or not: If the clause is satisfied, then  $X\cap P$ contains the items that represent the three literals of $C_{j,k}$ and the digit has value $3$. Otherwise, $X\cap P$ contains at most two of the items that represent the literals of  $C_{j,k}$ and the digit has value at most $2$. 

    Finally, for each for $i \in \{0,1,2,3\}$, we define the $w_{\phi_j}$-weight of item $a_{j,k,i}$ as $w_{\phi_j,a_{j,k,i}} = i \cdot 10^{k_j + k -1} + 10^{k-1}$. For all remaining items $i \in \mathcal{I} \setminus (X \cup A_{\phi_j})$, we define the partial $w_{\phi_j}$-weight to be zero. Furthermore, we define the partial capacity $B_{\phi_j} = \sum_{k=0}^{k_j-1} 10^k + \sum_{k=k_j}^{2k_j-1} 3 \cdot 10^k$. We claim that these definitions enforce~\Cref{prop:6}. 

    Intuitively, the fact that the $k_j$ decimal digits of lowest magnitude in $B_{\phi_j}$ have value $1$ forces a packing $P$ with $\sum_{i \in P} w_{i,\phi_j} = B_{\phi_j}$ to contain exactly one item of $\{a_{j,k,0},\ldots,a_{j,k,3}\}$ for each $k \in \{1,\ldots, k_j\}$ as each such item increases the corresponding digit $k-1$ by one. Similarly, the value of each digit $k_j+k-1$ in $B_{\phi_j}$ is three.
    Since the elements of $X$ that occur in clause $C_{j,k}$ increase the value of digit $k_j+k-1$ in $w_{\phi_j}(P)$ by one and item $a_{j,k,i}$ increases the digit by $i$, a packing $P$ with value three in digit $k_j+k-1$ of $w_{\phi_j}(P)$ can contain item $a_{j,k,0}$ if and only if $X \cap P$ contains the three items that correspond to the literals in $C_{j,k}$. This implies~\Cref{prop:6}. We give a more formal argumentation in~\Cref{app:hardness:of:approximation}.

    \paragraph*{Definition of set $\Phi$ and the $w_{\rho_j}$-weight:} Let $\Phi = \bigcup_{j \in \{1,\ldots,m\}} \Phi_j$ with $\Phi_j = \{y_j,u_j,f_{j,0},\ldots,f_{j,k_j-1}\}$. Note that $y_j$ is the item that has already been introduced for~\Cref{prop:major:3}. As a step toward enforcing this property, we define the $w_{\rho_j}$-weights such that:
    
    \begin{restatable}{property}{propSeven}
       \label{prop:7}
       For each $\phi_j$, a packing $P$ with $\sum_{i \in P} w_{i,\rho_j} = B_{\rho_j}$ has $y_j \in P$ if and only if $a_{j,k,0} \in P$ for some clause $C_{j,k}$ in $\phi_j$.
   \end{restatable}

    To enforce this property, we define the following partial capacity $B_{\rho_j} = k_j + 10^{k_j^2} + 10^{k_j^2+ 1}$. Next, we define the $w_{\rho_j}$-weight of the elements $a_{j,k,0}$, $k \in \{1,\ldots,k_j\}$, as $w_{a_{j,k,0},\rho_j} = 1$. Furthermore, we define $w_{u_j,\rho_j} = 10^{k_j^2 +1} + 10^{k_j^2} + k_j$, $w_{y_j,\rho_j} = 10^{k_j^2 +1}$ and $w_{f_{j,k},\rho_j} = 10^{k_j^2} + k$, for all $k \in \{0,\ldots,k_j-1\}$. For all other items, define the $w_{\rho_j}$-weight to be zero. We show in~\Cref{app:hardness:of:approximation} that these definitions imply~\Cref{prop:7}. 

   %  With these observations in place, we are ready to show that \Cref{prop:7} holds. To this end, fix a packing $P$ with $w_{\rho_j}(P) = B_{\rho_j}$. If $y_j \in P$, then $u_j \not\in P$ and, thus, $|P \cap \{f_{j,0},\ldots,f_{j,k_j-1}\}| = 1$. Let $f_{j,k}$ with $k \in \{0,\ldots, k_j-1\}$ be the element in $P \cap \{f_{j,0},\ldots,f_{j,k_j-1}\}$. The $w_{\rho_j}$-weight of the two items $y_j$ and $f_{j,k}$ is $w_{y_j,\rho_j}+w_{f_{j,k},\rho_j} = 10^{k_j^2+1} + 10^{k_j^2} + k$ for $k < k_j$. Thus, $P$ needs at most one more item with a positive $w_{\rho_j}$-weight to satisfy $w_{\rho_j}(P) = B_{\rho_j}$. Since we already argued that $u_j \not\in P$, this implies that $P$ must contain one element $a_{j,k',0}$ with $k' \in \{1,\ldots,k_j\}$. This gives us the first direction of minor property 7.

   %  For the second direction, assume $y_j \not\in P$. Then, we must have $u_j \in P$. Since $w_{u_j,\rho_j} = B_{\rho_j} = w_{\rho_j}(P) $, this implies that $P$ cannot contain any other item with positive $w_{\rho_j}$-weight. Hence, $P$ does not contain any $a_{j,k',0}$ with $k' \in \{1,\ldots,k_j\}$, which concludes the proof of minor property 7.

    \paragraph*{Definition of the uncertainty intervals and precise profits:} To finish the reduction, we define the profits and uncertainty intervals of all introduced items:
    \begin{itemize}
        \item For the items $y_j$, $j \in \{1,\ldots,m\}$, we define the uncertainty interval $I_{y_j}=(w_{y_j}-2,w_{y_j}+\epsilon)$ for a fixed $0 < \epsilon < \frac{1}{m}$. We define the profits as $p_{y_j} = w_{y_j}-1$.
        \item For all items $i \in \mathcal{I} \setminus \{y_j \mid j \in \{1,\ldots,m\}\}$, we use trivial uncertainty intervals $I_i = \{w_i\}$.
    \end{itemize}
    
    \paragraph*{Proof of the three main properties:} With the full construction in place, we are ready to prove the three main properties from the beginning of the proof:

    \begin{enumerate}
        \item \textbf{\Cref{prop:major:1}:} By definition of the profits, we have $p(P) \le w(P)$ for each packing $P$, which implies that the maximum profit is $p^* \le B$. Since the packing $P=L=\{i^*\}$ has a profit of exactly $B$, we get $p^* = B$. On the other hand, the upper limit $U_P$ of a packing $P$ is $w(P) + \epsilon |P \cap \{y_j \mid j \in \{1,\ldots,m\}|$ as only the items in $\{y_j \mid j \in \{1,\ldots,m\}$ have a non-trivial uncertainty interval with upper limits of $w_{y_j} + \epsilon$. By choice of $\epsilon$, this gives $U_P = w(P) + \epsilon |P \cap \{y_j \mid j \in \{1,\ldots,m\}| < w(P) + 1$. Since all weights are integer, this implies that $U_P \ge p^* = B$ only holds if $w(P) = B$. 
        \item \textbf{\Cref{prop:major:2}:} The property, in particular the first part, is essentially implied by~\Cref{prop:4} and~\Cref{prop:5}. We give the formal argumentation in~\Cref{app:hardness:of:approximation}.  
       %  By minor property 4, a packing $P$ satisfies $w(P) = B$ if and only if $w_x(P)= B_x$, $w_{\phi_j}(P) = B_{\phi_j}$ and $w_{\rho_j}(P) = B_{\rho_j}$ for all $j \in \{1,\ldots,m\}$. Then, \Cref{prop:5} implies that $w_x(P)= B_x$ if and only if $X \cap P$ represents a variable assignment. This gives us the first part of~\Cref{prop:major:2}. We show the second part, i.e., that each variable assignment is represented by some packing $P$ with $w(P) = B$, in~\Cref{app:hardness:of:approximation}.
        \item \textbf{\Cref{prop:major:3}:} By~\Cref{prop:major:1}, a packing $P$ has $U_P \ge p^*$ if and only if $w(P) = B$. By~\Cref{prop:4}, the latter holds if and only if $w_x(P)= B_x$, $w_{\phi_j}(P) = B_{\phi_j}$ and $w_{\rho_j}(P) = B_{\rho_j}$ for all $j \in \{1,\ldots,m\}$.        
        Fix a packing $P$ with $y_j \in P$ for some $j \in \{1,\ldots,m\}$. By \Cref{prop:7}, $y_j \in P$ holds if and only if $a_{j,k,0} \in P$ for some clause $C_{j,k}$ in $\phi_j$.  By \Cref{prop:6}, $a_{j,k,0} \in P$ if and only if $C_{j,k}$ is satisfied by the assignment represented by $X \cap P$. This gives the first part of \Cref{prop:major:3}. For the final part, observe that $U_{y_j} - p_{y_j} > 1$. On the other hand, $U_p-p^* < 1$. Hence, the second part of \Cref{prop:major:3} holds.
    \end{enumerate}

    To finish the proof of the reduction, it remains to argue about the running time and space complexity of the reduction. We do so in~\Cref{app:hardness:of:approximation}. The main argument is that, while the numerical values of the constructed weights are exponential, the number of digits in their decimal representations (and, thus, their encoding size) is polynomial.
\end{proof}

The previous theorem proves $\sum_2^p$-hardness for \offKU. Exploiting the inapproximability result  on the succinct set cover problem given in~\cite[Theorem 7.2]{Ta2021}, we can show the following stronger statement by using the same reduction.

\begin{restatable}{theorem}{hardnessApproximation}
   \label{thm:hardness:level2:approximation}
    Unless $\Sigma_2^p = \Delta_2^p$, there exists no $n_0^{1-\epsilon}$-approximation (given access to an oracle for problems in NP) for \offKU for any $\epsilon > 0$, where $n_0:= |\items \setminus \mathcal{I}_T|$.
\end{restatable}

\begin{proof}
    The reduction of~\Cref{thm:level2:hardness} satisfies the following: There is a solution for the succinct set cover instance of size $k$ if and only if there is a feasible query set of size $k$ for the constructed instance. Thus, the reduction is, in a sense, approximation factor preserving. 
    
    Furthermore, if $N$ is the encoding size of the input instance, then $n_0 \le N$ for the number of non-trivial items in the constructed instance. This directly follows from $N \ge m$ and $n_0 = m$. Note that $n_0=m$ holds by definition of the construction of~\Cref{thm:level2:hardness} since the reduction introduces exactly one non-trivial item $y_j$ for each formula $\phi_j$ and does not introduce any other non-trivial items.
    
    Let $f$ be a monotone non-decreasing function. Then we have $f(N) \ge f(n_0)$.  The observations above imply that if there is an $f(n_0)$-approximation (with access to an oracle deciding problems in NP) for the \offKU, then there is an $f(N)$-approximation (with access to an oracle deciding problems in NP) for succinct set cover.
    
    Since Ta-Shma et al.~\cite{Ta2021} showed that approximating succinct set cover within a factor of $N^{1-\epsilon}$ is $\sum_2^p$-hard for every $\epsilon > 0$, the theorem follows.
\end{proof}

\begin{remark}
    We remark that all results given in this section require large numerical input parameters. Hence, they do not prohibit the existence of pseudopolynomial algorithms.
\end{remark}

\section{Algorithmic Results}
\label{sec:offline:algorithms}

In this section, we give algorithms for $(\alpha,\beta)$-\offKU for different values of $\alpha$ and $\beta$. Motivated by the hardness results of the previous section, we show bounds on the size of the computed query sets in comparison to $|Q^*|$ instead of $|Q^*_{\alpha,\beta}|$. All our algorithms treat the two conditions on $(\alpha,\beta)$-feasible query sets (cf.~\Cref{def:alpha:beta:feasibility}) as two separate subproblems:
\begin{enumerate}
   \item Compute a query set $Q_1$ such that there exists a packing $P \subseteq Q_1 \cup \items_T$ with $p(P) \ge \frac{1}{\alpha} p^*$.
   \item Compute a query set $Q_2$ such that $U_P(Q_2) \le \beta \cdot p^*$ for all packings $P$.
\end{enumerate}
For our results, we solve these two problems and give bounds on $|Q_1 \cup Q_2|$ in terms of $|Q^*|$.

\subsection{The First Subproblem}

The following lemma solves the first subproblem for $\alpha = \frac{1}{1-\epsilon}$ by computing a packing $P$ with $p(P) \ge (1-\epsilon) \cdot p^*$ and $|P \setminus \items_T| \le |Q^*|$. The set $Q_1 = P\setminus \items_T$ satisfies the requirement of the subproblem and has $|Q_1| \le |Q^*|$. We prove the lemma by exploiting existing algorithms for two-dimensional knapsack.

\begin{restatable}{lemma}{lemSmallPacking}
    \label{lem:relax:small:packing}
    Fix an $\epsilon > 0$. Given an instance of \offKU, there exists a polynomial time algorithm that computes a packing $P$ with $p(P) \ge (1-\epsilon) \cdot p^*$ and $|P \setminus \items_T| \le |Q^*|$.
\end{restatable}

\begin{proof}
    We give the following algorithm and prove that it satisfies the lemma:
    \begin{enumerate}
        \item Let $\epsilon' = 1-\sqrt{1-\epsilon}$.
        \item For every integer $\ell$ from $1$ to $n$, we formulate the following two-dimensional knapsack problem:
        \begin{equation*}
            \tag{\ensuremath{\mathcal{P}_{\ell}}}
            \label{eq:2dim}
            \begin{array}{lll}
                \max &\sum_{i \in \items} y_i \cdot p_i\\
                \text{s.t. }& \sum_{i \in \items} y_i \cdot w_i \le B \\
                & \sum_{i \in \items\setminus \items_T} y_i \le \ell \\
                & y_i \in \{0,1\}& \forall i \in \items
            \end{array}
        \end{equation*}
        In this problem, the second constraint ensures that the selected packing contains at most $\ell$ non-trivial items. We can use the PTAS given in~\cite{frieze1984}, to compute a $(1-\epsilon')$-approximation for the two-dimensional knapsack instance. For each $\ell \in \{1,\ldots,n\}$, let $p'_\ell$ denote the profit of the computed solution $(1-\epsilon')$-approximation for~\eqref{eq:2dim}.
        \item Let $\ell^*$ denote the smallest integer that satisfies $p'_{\ell^*} \ge (1-\epsilon') \cdot p'_n$.
        \item Return the packing $P'_{\ell^*}$ that was computed for~\eqref{eq:2dim} with $\ell = \ell^*$.
    \end{enumerate}
    
    \textbf{Running time.} The running time of the algorithm is dominated by the $n$ executions of the PTAS of~\cite{frieze1984} in step $2$. Since the running time of the PTAS is polynomial in the input size, the running time of the algorithm is also polynomial in the input size. 
    
    \textbf{Correctness.} The profit $p'_n$ satisfies $p^* \ge p'_n \ge (1-\epsilon') \cdot p^*$ since~\eqref{eq:2dim} with $\ell = n$ is equivalent to the knapsack instance that is given as part of the offline problem as the second constraint is trivially satisfied for all packings.

    The profit $p'_{\ell^*}$ satisfies $p'_{\ell^*} \ge (1-\epsilon') \cdot p'_n$  by definition of the third step of the algorithm. 
    Thus,
    $$p'_{\ell^*} \ge (1-\epsilon') \cdot p'_n \ge (1-\epsilon') \cdot (1-\epsilon') \cdot p^* = (1-\epsilon) p^*.$$
    
    It remains to argue $\ell^* \le |Q^*|$. To this end, let $P^*$ denote a packing of minimum $|P^*\setminus \items_T|$ among all optimal packings for the given knapsack instance. This directly implies $|P^*| \le |Q^*|$ as $Q^*$ needs to contain $P^*\setminus \items_T$ all optimal packings $P^*$ to induce a feasible solution for~\eqref{ILP}.

    Furthermore, we know that a packing $P$ with $|P\setminus \items_T| \le \ell^*-1$ has profit at most $\frac{1}{1-\epsilon'} \cdot p'_{\ell^*-1}$ as $p'_{\ell^*-1}$ is a $(1-\epsilon')$-approximation~\eqref{eq:2dim} with $\ell = \ell^* - 1$. By definition of step 3 of the algorithm we also have $p'_{\ell^*-1} < (1-\epsilon')p'_n \le (1-\epsilon')p^*$. Thus, a packing with $\ell^*-1$ non-trivial items can have profit at most
 
    $$
    \frac{1}{1-\epsilon'} \cdot p'_{\ell^*-1} < \frac{1}{1-\epsilon'} \cdot (1-\epsilon')p^* = p^*.
    $$ 
    This implies $|Q^*| \ge |P^*\setminus \items_T| \ge \ell^*$ and concludes the proof.
 \end{proof}

\subsection{The Second Subprobem: Prefix Problems}

For the second subproblem, we consider special packings that correspond to greedy solutions for the underlying knapsack problem.  
For an item $i \in \items$, define the \emph{density} $\density_i = \frac{p_i}{w_i}$  and the \emph{optimistic density} $\odensity_i = \frac{U_i}{w_i}$. For a query set $Q \subseteq \items$, define the optimistic density of an item $i$ \emph{after} querying $Q$ as $\bar{d}_i(Q)  = \bar{d}_i$ if $i \not\in Q$ and  $\bar{d}_i(Q)  = d_i$ if $i \in Q$. We use $\bar{\prec}_Q$ to denote the \emph{optimistic density order} after querying $Q$, that is, for $i,j \in \items$ we use $i \bar{\prec}_Q j$ to denote that $\bar{d}_i(Q) \ge \bar{d}_j(Q)$. We assume an arbitrary but fixed tiebreaking rule between the items to treat $\bar{\prec}_Q$ as a total order. This allows us to define \emph{optimistic prefixes} and the \emph{prefix problem}.

\begin{definition}[Optimistic prefixes]
    For a set $Q \subseteq \items$ and a parameter $0 \le C \le B$, define the \emph{optimistic prefix} $\bar{F}_C(Q)$ to be the maximal prefix $S$ of order $\bar{\prec}_Q$ such that $w(S) \le C$. To shorten the notation, we define $\bar{F}(Q) := \bar{F}_B(Q)$
\end{definition}

From the analysis of the knapsack greedy algorithm, it is well-known that the following holds for all packings $P$:
\begin{equation}
   \label{eq:bound}
   U_P(Q) \le U_{\bar{F}(Q)}(Q) + \max_{i \in \items} U_i(Q).    
\end{equation}
If we compute a set $Q$, such that $U_{\bar{F}(Q)} \le \beta' p^*$ and $\max_{i \in \items} U_i(Q) \le \beta' \cdot p^*$, then $Q$ solves the second subproblem for $\beta = 2 \cdot \beta'$, which motivates the following problem.

\begin{definition}[Prefix problem]
    Given an instance of \offKU and a threshold parameter $D \ge p^*$,  where $p^*$ is the optimal profit of the underlying knapsack instance, the \emph{prefix problem} asks for the set $Q \subseteq \items$ of minimum cardinality such that $U_{\bar{F}(Q)}(Q) \le D$.
\end{definition}

Unfortunately, the prefix problem preserves the hardness of knapsack, even if the given threshold is larger than $p^*$ by a constant factor. 

\begin{restatable}{theorem}{thmPrefixHardness}
   \label{thm:prefix:hardness}
    The prefix problem is weakly NP-hard for every $D = c \cdot p^*$ with $c \ge 1$.
\end{restatable}

\begin{proof}
    We show the statement by reduction from the weakly NP-hard~\cite{GareyJ79} \emph{subset sum problem}, where we are given a set $A= \{a_1,\ldots, a_n\} \subseteq \mathbb{N}$ and an integer parameter $H$. The goal is to decide whether there exists a subset $S \subseteq A$ with $\sum_{a_i \in S} a_i = H$.
    Let $W = \sum_{a_i \in A} a_i$ and assume w.l.o.g.~that $H \le \frac{W}{2}$ (otherwise we can replace $H$ with $W-H$ to reach an equivalent problem that satisfies the inequality). Furthermore, assume $W \ge 3$, which is true for all non-trivial problem instances. 
 
    \paragraph*{Construction.} Given such an instance, we construct an instance of the prefix problem as follows:
    \begin{enumerate}
        \item For each $a_i \in A$, we construct a knapsack item $k_i$ with
        \begin{enumerate}
            \item $w_{k_i} = a_i$,
            \item $U_{k_i} = c \cdot a_i$ and $L_{k_i} = 0$,
            \item and $p_{k_i} = \epsilon \cdot a_i$ for a sufficiently small $\epsilon > 0$ with  $\frac{1}{W} > \epsilon$ and $\frac{c\cdot (W-H)}{W+H + 1} > \epsilon$.
        \end{enumerate}
        We refer to the items $k_1,\ldots,k_n$ as \emph{normal items} and define $N = \{k_1,\ldots,k_n\}$.
        \item Define the knapsack capacity as $B = 2\cdot W$.
        \item Introduce $n$ items $b_1,\ldots, b_n$ with 
        \begin{enumerate}
            \item $w_{b_j} = B - W + H + 1$,
            \item $U_{b_j} = c \cdot (W-H)$ and $L_{b_j} = 0$,
            \item and $p_{b_j} = W-H$.
        \end{enumerate}
        We refer to the items $b_1,\ldots,b_n$ as \emph{blocking items}.
    \end{enumerate}
 
    \paragraph*{Properties of the constructed instance.}
    In the constructed instance, the normal items $k_i$ have an optimistic density of $\odensity_{k_i} = \frac{U_{k_i}}{w_{k_i}} = \frac{c a_i}{a_i} = c$ and the blocking items $b_j$ have an optimistic density of $$\odensity_{b_j} = \frac{c\cdot (W-H)}{B-W+H + 1} = \frac{c\cdot (W-H)}{W+H + 1}  < c.$$ 
    Thus, the optimistic density order $\bar{\prec}_\emptyset$ starts with the normal items in some order followed by the blocking items in some order. This implies $\bar{F}(\emptyset) = \{k_1,\ldots, k_n\}$, as the normal items together have weight $w(N) = W$, which leaves no space for any blocking item within the knapsack capacity $B$. More precisely, the remaining capacity is $B-w(N)= B - W = W$, but a single blocking item needs capacity $B-W+H+1 = W+H+1 > W$. 
 
    By definition of the profits $p_{k_i}$ for the normal items $k_i$, the optimal packing with respect to the profits certainly packs at least one blocking item $b_j$. This is because the profit of a single blocking item is $W-H > W - \frac{W}{2} > 1$ (using the assumption that $W \ge 3$), while the profit of all normal items combined is $p(N) = n \cdot \epsilon \cdot w(N) = \epsilon \cdot W < 1$. 
    
    Since two blocking items have a combined weight of
    $$
    2 \cdot (B-W+H+1) = 2 \cdot B - 2 \cdot W + 2 \cdot H + 2 = 4W - 2W + 2 \cdot H + 2 > 2W + 2 > B,
    $$
    the optimal packing contains exactly one blocking item. The remaining space in the knapsack besides the blocking item is $B - (B-W+H+1) = W-H-1$. Thus, the optimal packing can fill the remaining space with normal items with a total weight of at most $W-H-1$. By the assumption that $D \le \frac{W}{2}$, the remaining space is at least $\frac{W}{2}-1$. We can assume without loss of generality that at least one normal item has a weight of at most $\frac{W}{2} - 1$. This is, because subset sum instances where all items have weight at least $\frac{W}{2}$ are trivial and, thus, can be excluded. This implies that the optimal packing with respect to the profits packs one blocking item and at least one normal item. Furthermore, the weight of the packed normal items is at most $W-H-1$. Therefore, the total profit $p^*$ of the optimal packing satisfies  
    $$
    W-H < p^* \le W-H + \epsilon \cdot (W-H-1) \le W-H + \epsilon \cdot W < W-H + 1.
    $$
    Note that the profits of the normal items are \emph{not} integer, hence the inequality $ W-H < p^* < W-H + 1$ is not a contradiction.
 
    Since $\bar{F}(\emptyset) = N$, we have $U_{\bar{F}(\emptyset)}(\emptyset) = \sum_{i = 1}^{n} U_{k_i}(\emptyset) = \sum_{i = 1}^{n} U_{k_i} = c \cdot W >  c \cdot (W-H+1) > c \cdot p^*$. This implies that every feasible solution $Q$ to the prefix problem instance must contain at least one normal item. By definition of the weights and profits, $k_i \in Q$ and $b_j \not\in Q$ for a normal item $k_i$ and a blocking item $b_j$ implies that 
 
    $$\odensity_{k_i}(Q) = \epsilon <  \frac{c\cdot (W-H)}{W+H + 1} = \odensity_{b_j}(Q).$$

    Hence, if we query $k_i$ but not $b_j$, the relative order between those items changes from  $\bar{\prec}_\emptyset$ to $\bar{\prec}_Q$.

    \paragraph*{Correctness.}
    To finish the proof, we show that there exists a feasible solution $Q$ for the constructed prefix problem instance with $Q < n$ if and only if there exists a subset $S\subseteq A$ with $\sum_{a_i \in S} a_i = H$.
 
    \textbf{First direction: } Assume there exists a subset $S\subseteq A$ with $\sum_{a_i \in S} a_i = H$. Consider the solution $Q = \{ k_i \in N \mid a_i \in S\}$ for the prefix problem. By assumption that $H \le \frac{W}{2}$, we have $|Q| = |S| < n$. It remains to argue that $Q$ satisfies $U_{\bar{F}(Q)}(Q) \le c \cdot p^*$. 
 
    Consider the normal items $N\setminus Q$. Since these items are part of $\bar{F}(\emptyset) = N$ and not part of $Q$, we have $N\setminus Q \subseteq \bar{F}(Q)$ as the optimistic density of these items is the same before and after querying $Q$ and the optimistic density of other items can only decrease by being queried. By assumption, we have $w(N\cap Q) = H$ and, therefore $w(N\setminus Q) = W-H$. Thus, the remaining space in $\bar{F}(Q)$ besides the items $N\setminus Q$ is 
    $$
    B- (W- H) = 2W - W + H = W + H.
    $$
    Since no blocking item is part of $Q$, the blocking items $b_j$ have the largest optimistic density besides the items $N\setminus Q$ (as we argued above). However, since $w_{b_j} = W + H + 1 > W +H$, no blocking items fits into the knapsack besides the items $N\setminus Q$. Thus, $\bar{F}(Q) = N\setminus Q$. This implies 
    $$
    U^F(Q) = U_{N\setminus Q}(Q) = c \cdot w(N\setminus Q) = c \cdot (W-H) < c \cdot p^*.
    $$
    We can conclude that $Q$ is feasible for the prefix problem instance.
 
    \textbf{Second direction:} Assume that there exists no subset $S \subseteq A$ with $\sum_{a_i \in S} a_i = H$. For the sake of contradiction, assume that $Q$ is a feasible solution to the constructed density prefix instance with $|Q| < n$. This implies that $Q$ contains neither all normal items nor all blocking items. By assumption, we have (i) $w(Q\cap N) < H$ or (ii) $w(Q\cap N) > H$.
 
    First, consider case (i). As argued in the first direction, the items $N \setminus Q$ must be part of $\bar{F}(Q)$. These items have a weight of $w(N\setminus Q) = w(N) - w(N\cap Q) > W - H$.  Since $W - H$ and $w(N\setminus Q)$ are integer, $w(N\setminus Q) > W - H$ implies $w(N\setminus Q) \ge W - H + 1$.
    Furthermore, the items satisfy $U_{N\setminus Q}(Q) = c \cdot w(N\setminus Q) \ge c \cdot (W-H+1)$. This implies 
    $$
    U_{\bar{F}(Q)}(Q) \ge U_{N\setminus Q}(Q) \ge c \cdot (W-H+1) > c \cdot p^*
    $$
    and $Q$ is not feasible for the prefix problem; a contradiction.
 
    Next, consider case (ii). In this case, the items $N\setminus Q$ have a weight $w(N\setminus Q) = W - w(N\cap Q) < W-H$. Since $W-H-1$ and $w(N\setminus Q)$ are integer, $w(N\setminus Q) < W-D$ implies $w(N\setminus Q) \le W-H-1$.
    Thus, the remaining space within the knapsack capacity besides the items $N\setminus Q$ is at least
    $$
    B - w(N\setminus Q) \ge 2W -W +H + 1 = W + H + 1.	
    $$
    By the assumption that $|Q| < n$, the item with the largest optimistic density besides $N\setminus Q$ after querying $Q$ is a blocking item $b_j \not\in Q$. By the calculation above, this blocking item still fits into the knapsack, i.e., $b_j \in \bar{F}(Q)$. As additionally $N\setminus Q \not= \emptyset$ by the assumption that $|Q| < n$, this implies
    $$
    U_{\bar{F}(Q)}(Q) \ge c \cdot (W-H) + c \cdot w(N\setminus Q) \ge c \cdot (W-H+1) > c \cdot p^*
    $$
    and, thus, $Q$ is not feasible for the prefix problem; a contradiction.
 \end{proof}

On the positive side, the problem can be solved to optimality in pseudopolynomial time.
 
\begin{restatable}{theorem}{thmPrefixPseudo}
    \label{thm:prefix:pseudo}
    The prefix problem can be solved in pseudopolynomial time.
\end{restatable}

We give the full proof of~\Cref{thm:prefix:pseudo} in~\Cref{app:algorithms}, but highlight the main ideas here.  In the following, we use $Q^*_F$ to refer to an optimal solution for the prefix problem. 

Assume for now, that the algorithm knows the last item $i_1$ in the prefix $\bar{F}(Q^*_F)$ and the first item  $i_2$ outside of $\bar{F}(Q^*_F)$ in the order $\bar{\prec}_{Q^*_F}$  (\Cref{fig1} (a)) and, for the sake of simplicity, assume that $i_1,i_2 \not\in Q^*_F$. We design our algorithm to reconstruct the optimal solution $\bar{F}(Q^*_F)$ (or a similar solution) using just the knowledge of $i_1$ and $i_2$.

\begin{figure}[t]
   \begin{tikzpicture}[scale = 0.75, transform shape]
       \node[] at (-5.5,0.5){(a)};

       \draw [draw=black,thick] (-5,0) rectangle ++(5,1);
       \node[] at (-2.5,0.5){$\bar{F}(Q_F^*) \setminus \{i_1\}$};

       \draw [draw=blue,thick] (0,0) rectangle ++(1.5,1);
       \node[] at (0.75,0.5){$\textcolor{blue}{i_1}$};

       \draw[red,thick,dashed] (2.5,1.5) -- (2.5,-0.5);
     
       \draw [draw=black,thick] (2.55+1.5,0) rectangle ++(7.5,1);
       \node[] at (2.55+1.5+3.75,0.5){$\items \setminus (\bar{F}(Q_F^*) \cup \{i_2\})$};

       \draw [draw=orange,thick] (2.55,0) rectangle ++(1.5,1);      
       \node[] at (2.55+0.75,0.5){$\textcolor{orange}{i_2}$};
       \node[] at (2.5,1.75){$\textcolor{red}{B}$};
   \end{tikzpicture}

   \begin{tikzpicture}[scale = 0.75, transform shape]
       \node[] at (-5.5,0.5){(b)};

       \draw [draw=black,thick] (4.5,0) rectangle ++(2.5,1);
       \node[] at (4.5+1.25,0.5){$A$};

       \draw [draw=black,thick] (8.5,0) rectangle ++(2,1);
       \node[] at (8.5+1,0.5){$R$};

       \draw [draw=black,thick] (-5,0) rectangle ++(8,1);
       \node[] at (-5+4,0.5){$S$};

       \draw [draw=blue,thick] (3,0) rectangle ++(1.5,1);
       \node[] at (3.75,0.5){$\textcolor{blue}{i_1}$};

       \draw [draw=orange,thick] (7,0) rectangle ++(1.5,1);      
       \node[] at (7+0.75,0.5){$\textcolor{orange}{i_2}$};

       \draw[red,thick,dashed] (2.5,1.5) -- (2.5,-0.5);
       \node[] at (2.5,1.75){$\textcolor{red}{B}$};

   \end{tikzpicture}

   \begin{tikzpicture}[scale = 0.75, transform shape]
       \node[] at (-5.5,0.5){(c)};

       \draw [draw=black,thick] (6,0) rectangle ++(4.5,1);
       \node[] at (6+2.25,0.5){$R \cup A$};

       \draw [draw=black,thick] (-5,0) rectangle ++(8,1);
       \node[] at (-5+4,0.5){$S$};

       \draw [draw=blue,thick] (3,0) rectangle ++(1.5,1);
       \node[] at (3.75,0.5){$\textcolor{blue}{i_1}$};

       \draw [draw=orange,thick] (4.5,0) rectangle ++(1.5,1);      
       \node[] at (4.5+0.75,0.5){$\textcolor{orange}{i_2}$};

       \draw[red,thick,dashed] (2.5,1.5) -- (2.5,-0.5);
       \node[] at (2.5,1.75){$\textcolor{red}{B}$};

   \end{tikzpicture}

   \begin{tikzpicture}[scale = 0.75, transform shape]
       \node[] at (-5.5,0.5){(d)};

       \draw [draw=black,thick] (2.55+1.5,0) rectangle ++(7,1);
       \node[] at (2.55+1.5+3.5,0.5){$R \cup A \cup S_1'$};

       \draw [draw=black,thick] (-5,0) rectangle ++(5.5,1);
       \node[] at (-5+2.25,0.5){$S\setminus S_1'$};

       \draw [draw=blue,thick] (0.5,0) rectangle ++(1.5,1);
       \node[] at (0.5+0.75,0.5){$\textcolor{blue}{i_1}$};

       \draw [draw=orange,thick] (2.55,0) rectangle ++(1.5,1);      
       \node[] at (2.55+0.75,0.5){$\textcolor{orange}{i_2}$};

       \draw[red,thick,dashed] (2.5,1.5) -- (2.5,-0.5);
       \node[] at (2.5,1.75){$\textcolor{red}{B}$};

   \end{tikzpicture}
   \caption{Illustration of the algorithmic ideas used to prove~\Cref{thm:prefix:pseudo}.}
   \label{fig1}
\end{figure}
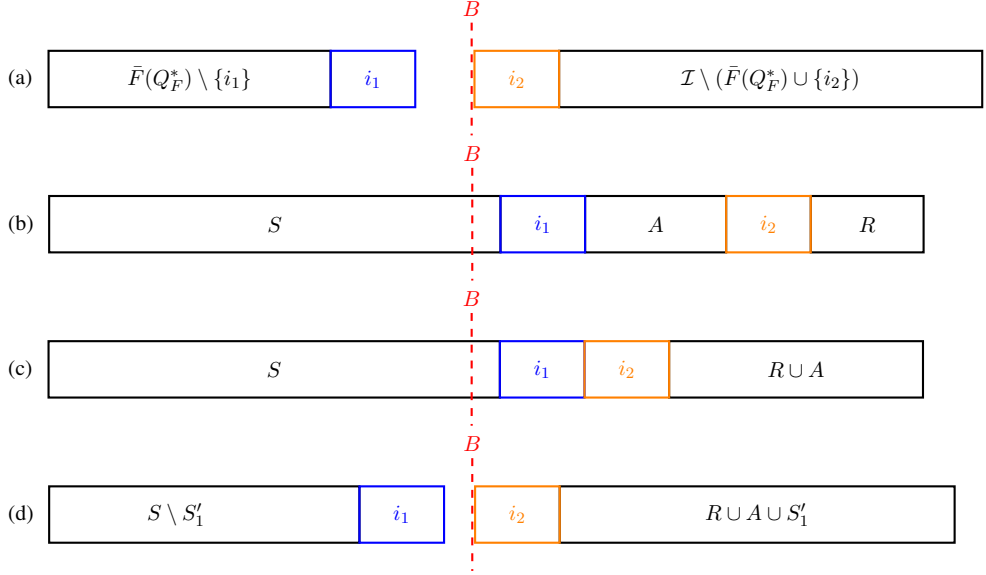

If we look at the same two items $i_1,i_2$ in the initial optimistic density order $\bar{\prec}_\emptyset$, then there can be a subset of items $S$ before $i_1$, a subset of items $A$ between $i_1$ and $i_2$ and a subset of items $R$ after $i_2$ (\Cref{fig1} (b)).
Based on $i_1$ and $i_2$ being next to each other in order $\bar{\prec}_{Q^*_F}$, we can immediately deduce that $A \subseteq Q^*_F$ as items in $A \setminus Q_F^*$ would still be between $i_1$ and $i_2$ in $\bar{\prec}_{Q^*_F}$. Thus, the algorithm can safely add $A$ to its solution as the optimal solution does the same. Similarly, from $i_2$ being the first item outside of $\bar{F}(Q^*_F)$, it is clear that $R \cap Q_F^* = \emptyset$ as the items in $R$ stay outside the prefix $\bar{F}(Q^*_F)$ whether they are queried or not. Hence, the algorithm can safely ignore such items.

In the order $\bar{\prec}_A$, i.e., in the order after adding $A$ to the solution, the items $i_1$ and $i_2$ must already be next to each other. However, we still can have $i_1 \not\in \bar{F}(A)$, that is, $i_1$ might not yet be part of the prefix (\Cref{fig1} (c)). To fix this, the algorithms needs to query items from the set $S_1 \subseteq S$, which contains items that are before $i_1$ in the order $\bar{\prec}_A$ but would move behind $i_2$ if they are added to the solution. To reconstruct the optimal solution, the algorithm has to query these items in such a way that $i_1$ enters the prefix but $i_2$ does not. On the one hand, the algorithm should select items $i \in S_1$ with high $U_i$ as the items will leave the prefix and the goal of the prefix problem is to decrease the upper limit of the prefix below the threshold $D$. On the other hand, the algorithm needs to make sure not to query too many items in $S_1$, so the cardinality of the solution does not grow too large. If $K$ is the minimum amount of weight that needs to be queried for $i_1$ to enter the prefix, $D$ is the maximum amount of weight that can be queried before $i_2$ enters the prefix, and $n_1$ is the number of queries that the algorithm can afford, then the algorithm should select its queries from $S_1$ according to the following ILP, which we show to be solvable in pseudopolynomial time:
\begin{equation}
   \tag{\ensuremath{\mathcal{P}_{S_1}}}\label{density:sub:ILP}
   \begin{array}{lll}
       \max &\sum_{i \in S_1} x_i \cdot U_i\\
       \text{s.t. }& \sum_{i \in S_1} x_i \cdot w_i \ge K & \\
       & \sum_{i \in S_1} x_i \cdot w_i \le H & \\
       & \sum_{i \in S_1} x_i = n_1 & \\
       & x_i \in \{0,1\}& \forall i \in S_1
   \end{array}
\end{equation}

After adding the solution $S_1'$ of~\eqref{density:sub:ILP} to the solution, the prefix $\bar{F}(A \cup S_1')$ of the algorithm has reached roughly the same configuration as $\bar{F}(Q^*_F)$: $i_1$ is last in the prefix and $i_2$ is first outside the prefix (\Cref{fig1} (d)). However, we still might have $U_{\bar{F}(A \cup S_1')}(A \cup S_1') > D$. To fix this, the algorithm has to query items of $S \setminus S_1'$ that stay in front of $i_1$ in the optimistic density order even after being queried. Since these items are part of the prefix $\bar{F}(A \cup S_1')$ \emph{and} will stay part of the prefix even after being queried, the algorithm should greedily query such items $i$ with large $U_i - p_i = U_{\bar{F}(A \cup S_1')}(A \cup S_1') - U_{\bar{F}(A \cup S_1' \cup \{i\})}(A \cup S_1' \cup \{i\})$ until the solution becomes feasible for the prefix problem instance.
Our algorithm for~\Cref{thm:prefix:pseudo} is based on exactly this approach. We show in~\Cref{app:algorithms} how to formalize this and get rid of all assumptions that were used in the description above.

To achieve polynomial running time, we use the same approach but instead of~\eqref{density:sub:ILP} we solve a certain LP-relaxation with the property that an optimal basic feasible solution has at most two fractional variables. Omitting the fractional items leads to the following corollary. 

\begin{restatable}{corollary}{thmPrefixNonPseudo}
    \label{thm:prefix:nonpseudo}
    Given an instance of the prefix problem with threshold parameter $D$, let $Q^*_F$ denote an optimal solution to the instance. There exists a polynomial time algorithm that computes a set $Q$ with $|Q| \le |Q^*_F|$ such that $U_{\bar{F}(Q)}(Q) \le D + 2 \cdot \max_{i \in \items} U_i$.
\end{restatable}

\subsection{Combining the Subproblems}

By combining~\Cref{lem:relax:small:packing} and~\Cref{thm:prefix:pseudo}, we can show the following theorem. The idea is to use~\Cref{lem:relax:small:packing} to compute a packing $P$ with $p(P) \ge (1-\epsilon') p^*$ for a carefully chosen $\epsilon'>0$ and then use~\Cref{thm:prefix:pseudo} with threshold $D = \frac{p(P)}{(1-\epsilon')}$ to compute a solution $Q'$ to the prefix problem. Exploiting~\eqref{eq:bound}, we return the solution $Q = (P \setminus \items_T) \cup Q' \cup \{i \in \items \mid U_i > D\}$ and observe that $Q$ satisfies the theorem. 

\begin{restatable}{theorem}{thmBlackboxPseudo}
    \label{thm:2blackbox:pseudo}
    Fix $\epsilon > 0$. There exists a pseudopolynomial algorithm that given an instance of \offKU computes a $(\frac{1}{1-\epsilon},(1+\epsilon)\cdot 2)$-feasible query set $Q$ with $|Q| \le 2 |Q^*|$.
\end{restatable}

\begin{proof}
    We give the following algorithm and show that it satisfies the requirements of the theorem:
    \begin{enumerate}
        \item Initialize $Q = \emptyset$.
        \item Fix $\epsilon' = \frac{\epsilon}{\epsilon+1}$ and initialize $Q= \emptyset$. Note that $\epsilon' \le \epsilon$ and $(1+ \epsilon) = \frac{1}{1-\epsilon'}$.
        \item Use~\Cref{lem:relax:small:packing} to find a packing $P$ with $p(P) \ge (1-\epsilon') p^*$ and $|P \setminus \items_T| \le |Q^*|$. Add $P\setminus \items_T$ to $Q$, define $D = \frac{p(P)}{1-\epsilon'}$ and note that $D \ge p^*$.    
        \item Add all items $i$ with $U_i > D$ to $Q$.
        \item Replace the intervals of all $i \in Q$ with $I_i=\{p_i\}$. Use the algorithm of~\Cref{thm:prefix:pseudo} to solve the prefix problem with threshold $D$ on the resulting instance to compute the minimum cardinality query set $Q'$ subject to $U_{\bar{F}(Q')}(Q') \le D$. Add $Q'$ to $Q$.
        \item Return the set $Q$.
    \end{enumerate}
 
    The running time of this algorithm is dominated by the running time of~\Cref{thm:prefix:pseudo} and, thus, pseudopolynomial. It remains to show that $Q$ is $(\frac{1}{1-\epsilon},(1+\epsilon)2)$-feasible and $|Q| \le |Q^*|$.

    \textbf{Bound on $|Q|$.} For the latter, we can first observe that $|P \setminus \items_T| \le Q^*$ by~\Cref{lem:relax:small:packing}. Then, we can observe that $i \in Q^*$ for all items $i$ with $U_i > D$. Otherwise, $U_P > D \ge p^*$ for the packing $P=\{i\}$, which contradicts the feasibility of $Q^*$. Finally, we can observe that $\bar{Q}^* := Q^* \setminus \left(\{i \in \items \mid U_i > D\} \cup (P \setminus \items_T)\right)$ must be a feasible solution to the prefix problem of step 3. Otherwise, $U_{\bar{F}(Q^*)}(Q^*) > D \ge p^*$  which contradicts the feasibility of $Q$. Since $\bar{Q}^*$ is a feasible solution for the prefix problem, we must have $|Q'| \le |\bar{Q}^*|$ as $Q'$ is an optimal solution for the prefix problem. Hence, $|Q| = |\{i \in \items \mid U_i > D\}| + |Q'| + |P \setminus \items_T| \le |Q^*| + |P \setminus \items_T|  \le 2 \cdot |Q^*|$.
 
    \textbf{Feasibility of $Q$.} We start by showing the first condition of~\Cref{def:alpha:beta:feasibility}. To this end, note that~\Cref{lem:relax:small:packing} and the choice of $\epsilon' \le \epsilon$  imply
    
    $$
    p(P) \ge (1- \epsilon') p^* \ge (1-\epsilon) p^*,
    $$
 
    which immediately implies the first condition of~\Cref{def:alpha:beta:feasibility}. For the second condition, observe that $Q'$ being a solution to the prefix problem with threshold $D$ for the instance where $Q\setminus Q'$ has already been queried (this is the reason we modified the uncertainty intervals in step 5) implies $U_{\bar{F}(Q)}(Q) \le D$. Furthermore, since $Q$ contains all elements $i$ with $U_i > D$ and $p_i \le p^*$ holds by assumption that $w_i \le B$, we can observe that $\max_{i \in \items} U_i(Q) \le D$. Hence,\eqref{eq:bound} implies
 
    $$
        U_P(Q) \le U_{\bar{F}(Q)}(Q) +\max_{i \in \items} U_i(Q) \le 2D
    $$
 
    for all packings $P$. We can finish the proof by observing $D \le \frac{1}{1-\epsilon'} p(P) \le \frac{1}{1-\epsilon'} p^* = (1+ \epsilon) p^*$.
 \end{proof}

 Replacing the usage of~\Cref{thm:prefix:pseudo} with~\Cref{thm:prefix:nonpseudo} in the approach above yields:

\begin{restatable}{theorem}{thmBlackboxNonPseudo}
    \label{thm:2blackbox}
    Fix $\epsilon > 0$. There exists a polynomial time algorithm that given an instance of \offKU computes a $(\frac{1}{1-\epsilon},(1+\epsilon)\cdot 4)$-feasible query set $Q$ with $|Q| \le 2 |Q^*|$.
\end{restatable}

\begin{proof}
    We can use the same algorithm as in~\Cref{thm:2blackbox:pseudo} but replace the the usage of~\Cref{thm:prefix:pseudo} wit~\Cref{thm:prefix:nonpseudo}. In the proof, we can just replace the bound $U_{\bar{F}(Q)}(Q) \le D$ with the weaker $U_{\bar{F}(Q)}(Q) \le D + 2 \cdot \max_{i \in \items}  U_i(Q) \le 3 D$, where the last inequality uses that all items $i$ with $U_i > D$ have already been queried before the prefix problem is solved.
 \end{proof}

\section{Conclusion}

We hope that our results on \offKU improve the understanding of NP-hard problems under explorable uncertainty. In particular, our algorithmic insights on the resource augmentation setting give hope for tackling such problems even if the corresponding offline versions have strong impossibility results. For knapsack specifically, studying a stochastic version of the prefix problem of~\Cref{sec:offline:algorithms}, for example in the stochastic setting of~\cite{MegowS23}, seems like a logical next step towards algorithmic results for the non-offline~\KU. For \offKU, our results of~\Cref{sec:offline:algorithms} show that non-trivial algorithmic results with theoretical guarantees are possible, opening the door for more research on finding the best-possible guarantees.

\bibliography{arxiv}

\appendix

\section{Missing parts from the proof of~\Cref{thm:level2:hardness}}
\label{app:hardness:of:approximation}

We continue by giving the omitted parts from the proof of~\Cref{thm:level2:hardness}. First, we show that $(\alpha,\beta)$-\offKU is in $\Sigma_2^p$.

\begin{restatable}{lemma}{lemSigmaMembership}
    \label{lem:sigma:membership}
    The $(\alpha,\beta)$-\offKU is in $\Sigma_2^p$.
\end{restatable}
\begin{proof}
    We have to show that the $(\alpha,\beta)$-\offKU can be solved in non-deterministic polynomial time if we have access to an oracle that decides NP-complete problems.
    
    Consider an instance of the decision problem variant, where we are given an instance of the $(\alpha,\beta)$-\offKU for some $\alpha,\beta \ge 1$, and the goal is to decide whether there exists an $(\alpha,\beta)$-feasible query set $Q$ with $|Q| \le k$.
    
    We give a certificate of polynomial size that, given access to the oracle for problems in NP, can be verified in polynomial time.
    As certificate, we use a query set $Q \subseteq \items$ with $|Q| \le k$. 
    It remains to argue that we can verify in polynomial time, whether $Q$ is $(\alpha,\beta)$-feasible. That is, we have to verify the following two conditions of~\Cref{def:alpha:beta:feasibility}:
    \begin{enumerate}
        \item There is a packing $P$ such that $P \subseteq Q \cup \items_T$ and $p(P) \ge \frac{1}{\alpha} \cdot p^*$.
        \item $U_P(Q) \le \beta \cdot p^*$ for every packing $P$.    
    \end{enumerate}

    To verify that $Q$ satisfies the first condition, we can first use the oracle for deciding problems in NP to compute the value of $p^*$ via binary search.  Afterwards, we can compute the optimal knapsack profit $p'$ for the subinstance that only contains the items of $P \subseteq Q \cup \items_T$ in the same way. The first condition is satisfied if and only if $p' \ge \frac{1}{\alpha} p^*$.

    For the second condition, we can again compute $p^*$ in the same way as before. 
    Afterwards, let $U^*(Q)$ denote the optimal profit for the knapsack instance that uses profits $p'_i = U_i(Q)$ instead of the original profits $p_i$. We can compute the value of $U^*(Q)$ in the same way as $p^*$. The second condition is satisfied if and only if $U^*(Q) \le \beta \cdot p^*$.
\end{proof}

 Next, we move on to prove that if the instance constructed by the reduction of~\Cref{thm:level2:hardness} satisfies \Cref{prop:major:1}, \Cref{prop:major:2} and \Cref{prop:major:3}, then the reduction is correct.

 \begin{claim}
   Assume the constructed instance satisfied \Cref{prop:major:1}, \Cref{prop:major:2} and \Cref{prop:major:3}.
   Then, there exists a solution $S$ with $|S| \le k$ for the given succinct set cover instance if and only if there is a feasible query set $Q$ for the constructed \offKU instance with $|Q| \le k$.
 \end{claim}

 \begin{proof}
   Since the first direction of this statement has already been shown in the main part, we proceed by showing the second direction.

   \textbf{Second direction.} Assume $Q$ with $|Q| \le k$ is feasible for the constructed \offKU instance. As the items $y_j$ for the formulas $\phi_j$ are the only non-trivial items, we have $Q \subseteq Y = \{y_1,\ldots,y_m\}$. We claim that $S = \{j \mid y_j \in Q\}$ is feasible for the succinct set cover instance. 
    
   To this end, consider an arbitrary variable assignment $\varphi$. By~\Cref{prop:major:2}, there is at least one packing $P$ with $w(P) = B$ and $U_P  > P^*$ such that $X \cap P$ represents $\varphi$. \Cref{prop:major:3} implies that $P \cap Y$ contains exactly the elements $y_j$ such that $\phi_j$ is satisfied by $\varphi$.
   By assumption of the problem, $\varphi$ satisfies at least one formula $\phi_j$ and, thus, $P \cap Y \not= \emptyset.$ 
   Since $Q$ is feasible, we have $Q \cap P \not= \emptyset$. Hence, there is at least one $y_j \in Q \cap P$ and $\varphi$ satisfies $\phi_j$. By definition of $S$, it contains the index $j$. Thus, $\varphi \in \bigcup_{j \in S} S_j$. As this holds for all variable assignments, we have that $S$ is a feasible solution to the succinct set cover instance.
 \end{proof}

 Next, we show that the constructed instance satisfies \Cref{prop:6} and \Cref{prop:7}.

 \propSix*

 Before we show that the property is satisfied, we repeat the relevant parts from the reduction: 
 
 \begin{quote}
   To achieve the property, we first define the partial $w_{\phi_j}$-weights for the items $X$.
   For a literal $x_i$, let $\mathcal{C}_{x_i,j} = \{ k \mid x_i \text{ occurs in } C_{j,k}\}$ and define  $\mathcal{C}_{\neg x_i,j}$ in the same way.
   We define the weights $w_{v_i, \phi_j}$ and $w_{\bar{v}_i, \phi_j}$ as $\sum_{k \in \mathcal{C}_{x_i,j}} 10^{k_j + k - 1}$ and $\sum_{k \in \mathcal{C}_{\neg x_i,j}} 10^{k_j + k - 1}$, respectively.
   Intuitively, digit $k_j+k$ of the sum $\sum_{h \in X \cap P} w_{h,\phi_j}$ of a packing $P$ with $w(P) = B$ indicates whether the assignment represented by $X\cap P$ satisfies clause $C_{j,k}$ or not: If the clause is satisfied, then  $X\cap P$ contains the items that represent the three literals of $C_{j,k}$ and the digit has value $3$. Otherwise, $X\cap P$ contains at most two of the items that represent the literals of  $C_{j,k}$ and the value is at most $2$. 

   Finally, for each for $i \in \{0,1,2,3\}$, we define the $w_{\phi_j}$-weight of item $a_{j,k,i}$ as $w_{\phi_j,a_{j,k,i}} = i \cdot 10^{k_j + k -1} + 10^{k-1}$. For all remaining items $i \in \mathcal{I} \setminus (X \cup A_{\phi_j})$, we define the partial $w_{\phi_j}$-weight to be zero. Furthermore, we define the partial capacity $B_{\phi_j} = \sum_{k=0}^{k_j-1} 10^k + \sum_{k=k_j}^{2k_j-1} 3 \cdot 10^k$. We claim that these definitions enforce \Cref{prop:6}. 
 \end{quote}

 \begin{proof}[Proof of~\Cref{prop:6}]
   To see this, fix a packing $P$ with  $\sum_{i \in P} w_{i,\phi_j} = B_{\phi_j}$ and $\sum_{i \in P} w_{i,x} = B_x$, and consider the value $w_{\phi_j}(P) := \sum_{i \in P} w_{i,\phi_j} = \sum_{i \in P\cap (X\cup A_{\phi_j})} w_{i,\phi_j}$. By definition of the $w_{\phi_j}$-weights, the $k_j$ digits of lowest magnitude in the decimal representation of $w_{\phi_j}(P)$ can be between $\sum_{k=0}^{k_j-1} 4 \cdot 10^k$ (if $A_{\phi_j} \subseteq P$) and $0$ (if $A_{\phi_j} \cap P = \emptyset$). Note that items outside of $A_{\phi_j}$ do not have any influence on these $k_j$ digits of $w_{\phi_j}(P)$. For $P$ to satisfy $\sum_{i \in P} w_{i,\phi_j} = B_{\phi_j}$, we need to have that the $k_j$ digits of lowest magnitude in $w_{\phi_j}(P)$ have value $\sum_{k=0}^{k_j-1} 10^k$. The latter is the case if any only if $|\{a_{j,k,0},a_{j,k,1},a_{j,k,2},a_{j,k,3}\} \cap P| = 1$ for all $k \in \{1,\ldots,k_j\}$.

   In a similar way, the $k_j$ digits of highest magnitude in $w_{\phi_j}(P)$ all have a value between $0$ and $9$. The digit $k_j+k-1$ has value $9$ if and only if $\{a_{j,k,1},a_{j,k,2},a_{j,k,3}\} \subseteq P$ and $P$ contains the three items of $X$ that correspond to the literals in clause $C_{j,k}$. Since we already argued $|\{a_{j,k,0},a_{j,k,1},a_{j,k,2},a_{j,k,3}\} \cap P| = 1$, the maximum value of digit $k_j+k-1$ is only $7$. In a packing $P$ with $w_{\phi_j}(P) = B_{\phi_j}$, the value of each such digit has to be exactly $3$. Thus, $a_{j,k,0} \in P$ if and only if $X\cap P$ contains the three items that represent the literals in clause $C_{j,k}$, which is the case if and only if the assignment represented by $X\cap P$ satisfies clause $C_{j,k}$. We conclude that \Cref{prop:6} is satisifed.
 \end{proof}

 \propSeven*

 Before we show that the property is satisfied, we repeat the relevant parts from the reduction:
 
 \begin{quote}
   To enforce this property, we define the following partial capacity $B_{\rho_j} = k_j + 10^{k_j^2} + 10^{k_j^2+ 1}$. Next, we define the $w_{\rho_j}$-weight of the elements $a_{j,k,0}$, $k \in \{1,\ldots,k_j\}$, as $w_{a_{j,k,0},\rho_j} = 1$. Furthermore, we define $w_{u_j,\rho_j} = 10^{k_j^2 +1} + 10^{k_j^2} + k_j$, $w_{y_j,\rho_j} = 10^{k_j^2 +1}$ and $w_{f_{j,k},\rho_j} = 10^{k_j^2} + k$, for all $k \in \{0,\ldots,k_j-1\}$. For all other items, define the $w_{\rho_j}$-weight to be zero. 
 \end{quote}

 \begin{proof}[Proof of~\Cref{prop:7}]
   Consider a packing $P$ with $w_{\rho_j}(P) = B_{\rho_j}$. First, observe that $|P \cap \{y_j,u_j\}| = 1$ as otherwise we either have $w_{\rho_j}(P) < 10^{k_j^2+1} \le B_{\rho_j}$ or $w_{\rho_j}(P) \ge 2 \cdot 10^{k_j^2+1} > B_{\rho_j}$. In a similar way, we can argue that a packing $P$ with $w_{\rho_j}(P) = B_{\rho_j}$ must satisfy $|P \cap \{u_j,f_{j,0},\ldots,f_{j,k_j-1}\}| = 1$. Otherwise, given that we already know that $|P \cap \{y_j,u_j\}| = 1$, we either have $w_{\rho_j}(P) < 10^{k_j^2+1} + 10^{k_j^2+1} \le B_{\rho_j}$ or $w_{\rho_j}(P) \ge 10^{k_j^2+1} + 20^{k_j^2+1} > B_{\rho_j}$. 

   With these observations in place, we are ready to show that \Cref{prop:7} holds. To this end, fix a packing $P$ with $w_{\rho_j}(P) = B_{\rho_j}$. If $y_j \in P$, then $u_j \not\in P$ and, thus, $|P \cap \{f_{j,0},\ldots,f_{j,k_j-1}\}| = 1$. Let $f_{j,k}$ with $k \in \{0,\ldots, k_j-1\}$ be the element in $P \cap \{f_{j,0},\ldots,f_{j,k_j-1}\}$. The $w_{\rho_j}$-weight of the two items $y_j$ and $f_{j,k}$ is $w_{y_j,\rho_j}+w_{f_{j,k},\rho_j} = 10^{k_j^2+1} + 10^{k_j^2} + k$ for $k < k_j$. Thus, $P$ needs at most one more item with a positive $w_{\rho_j}$-weight to satisfy $w_{\rho_j}(P) = B_{\rho_j}$. Since we already argued that $u_j \not\in P$, this implies that $P$ must contain one element $a_{j,k',0}$ with $k' \in \{1,\ldots,k_j\}$. This gives us the first direction of~\Cref{prop:7}.
  
   For the second direction, assume $y_j \not\in P$. Then, we must have $u_j \in P$. Since $w_{u_j,\rho_j} = B_{\rho_j} = w_{\rho_j}(P) $, this implies that $P$ cannot contain any other item with positive $w_{\rho_j}$-weight. Hence, $P$ does not contain any $a_{j,k',0}$ with $k' \in \{1,\ldots,k_j\}$, which concludes the proof of~\Cref{prop:7}.
 \end{proof}

 Next, we discuss the combination of the partial weights and capacities into single weights and a single capacity such that~\Cref{prop:4} is satisfied.

 \propFour*

 The simple idea to combine the partial weights and capacities is to just concatenate the decimal representations of these number and introduce sufficiently many zeros between neighboring numbers in the concatenation to avoid \enquote{overflows} when summing multiple weights. To this end, let $d_x$ denote the number of digits in the decimal representation of $w_x(\items)$, for each $j \in \{1,\ldots,m\}$ let $d_{\phi_j}$ denote the number of digits in the decimal representation of $w_{\phi_j}(\items)$ and $d_{\phi_j}$ denote the number of digits in the decimal representation of $w_{\rho_j}(\items)$. Furthermore, let $D_{\phi_j} := \sum_{j \in \{1,\ldots,m\}} d_{\phi_j}$ and let $D_{\rho_j} := D_{\phi_j} + \sum_{j \in \{1,\ldots,m\}} d_{\rho_j}$.

 We define the combined capacity as
 $$
  B:= \sum_{j \in \{1,\ldots,m\}} B_{\phi_j} \cdot 10^{\sum_{j' < j} d_{\phi_j}} + \sum_{j \in \{1,\ldots,m\}} B_{\rho_j} \cdot  10^{ D_{\phi_j} + \sum_{j' < j} d_{\rho_j}} + B_x \cdot 10^{D_{\rho_j}}.
 $$

 For an item $i$, we define the combined weight $w_i$ as

 $$
   w_i :=  \sum_{j \in \{1,\ldots,m\}} w_{i,\phi_j} \cdot 10^{\sum_{j' < j} d_{\phi_j}} + \sum_{j \in \{1,\ldots,m\}} w_{i,\rho_j} \cdot  10^{ D_{\phi_j} + \sum_{j' < j} d_{\rho_j}} + w_{i,x} \cdot 10^{D_{\rho_j}}.
 $$

 Using these combined weights and capacities ensures~\Cref{prop:4}. Furthermore, the weights and capacities have a polynomial encoding size, even though their numerical values are exponential in the input size of the given succinct set cover instance.

 Finally, we give the missing part of the proof of~\Cref{prop:major:2}.

 \propTwo*

 \begin{proof}
   We start by showing the first part of the property, i.e., if a packing $P$ has $w(P) = B$ and $U_P > p^*$, then $X \cap P$ represents a variable assignment.

   By~\Cref{prop:4}, a packing $P$ satisfies $w(P) = B$ if and only if $w_x(P)= B_x$, $w_{\phi_j}(P) = B_{\phi_j}$ and $w_{\rho_j}(P) = B_{\rho_j}$ for all $j \in \{1,\ldots,m\}$. Then, \Cref{prop:5} implies that $w_x(P)= B_x$ if and only if $X \cap P$ represents a variable assignment. This gives us the first part of~\Cref{prop:major:2}. 

   Next, we show the second part of the property. To this end, consider a variable assignment $\varphi$. We construct a packing $P$ such that $w(P) = B$, $U_P > p^*$ and $X \cap P$ represents $\varphi$:
   \begin{enumerate}
       \item Start with $P = \emptyset$.
       \item For each $i \in \{1, \ldots,n\}$. If $x_i$ has value one in $\varphi$, add $v_i$ to $P$. Otherwise, add $\bar{v}_i$ to $P$. Note that this ensures that $w_x(P) = B_x$ and that $P \cap X$ represents $\varphi$.
       \item For each formula $\phi_j$ and each clause $C_{j,k}$ of $\phi_j$, let $h$ denote the number of literals in $C_{j,k}$ that are satisfied by $\varphi$. Add the item $a_{j,k,3-h}$ to $P$. Note that, by construction, this ensures $w_{\phi_j}(P) = B_{\phi_j}$.
       \item Finally, for each formula $\phi_j$, let $h$ denote the number of clauses that are satisfied by $\varphi$. If $h = 0$, then add $u_j$ to $P$. Otherwise, add $y_j$ and $f_{j,k_j-h}$ to $P$.  Note that, by construction, this ensures $w_{\rho_j}(P) = B_{\rho_j}$. 
   \end{enumerate}

   Since $w_x(P) = B_x$, $w_{\phi_j}(P) = B_{\phi_j}$ and $w_{\rho_j}(P) = B_{\rho_j}$ for all $j \in \{1,\ldots,m\}$, \Cref{prop:4} implies $w(P) = B$. Furthermore, by assumption that $\varphi$ satisfies at least one formula $\phi_j$, we have $P \cap \{y_1, \ldots, y_m\} \not= \emptyset$. Hence, $U_P > p^*$.
 \end{proof}

\paragraph*{Size and Running Time.} The sets $X, A, \Phi$ and $L$ all have size polynomial in the size of the input instance and, thus, can be constructed in polynomial size and space. The number of digits of all constructed weights are also polynomial in the input size. Thus, the encoding size of these weights is polynomial, although the numerical values are exponential in the input size. Therefore, the instance can be constructed in polynomial time and space. We remark that the numerical values of the weights are exponential in the the input size. 

\section{Missing Proofs of~\Cref{sec:offline:algorithms}}
\label{app:algorithms}

\thmPrefixPseudo*

\begin{proof}
   To prove the theorem, we give the following algorithm and show that it is optimal and has pseudopolynomial running time.
   \begin{enumerate}
       \item Initialize $Q = \emptyset$.
       \item Compute the following guesses\footnote{We use the term \enquote{guess} to express that we want to brute-force certain values. E.g., if we say that the algorithm guesses $i_1,i_2 \in \items$ then this means that the algorithms tries all tuples $(i_1,i_2) \in \items \times \items$} about the smallest set $Q_F^*$ that satisfies $U_{\bar{F}(Q_F^*)} \le c \cdot p^*$:
       \begin{enumerate}
           \item Guess the item $i_1 \in \bar{F}(Q_F^*)$ that has the smallest optimistic density $\odensity_{i_1}(Q_F^*)$ in $\bar{F}(Q_F^*)$, i.e., the last item in the prefix $\bar{F}(Q_F^*)$ in the $\bar{\prec}_Q$-order. Furthermore, guess whether $i_1 \in Q_F^*$ or $i_2 \not\in Q^*$. If the guess is $i_1 \in Q$, then add $i_1$ to $Q$. This leads to a total of $\mathcal{O}(n)$ guesses.
           \item Guess the item $i_2 \in \items \setminus \bar{F}(Q_F^*)$ that has the largest optimistic density $\odensity_{i_2}(Q_F^*)$ in $\items \setminus \bar{F}(Q_F^*)$, i.e., the item that comes first among the items of $\items\setminus \bar{F}(Q_F^*)$ in the $\bar{\prec}_{Q_F^*}$-order. Furthermore, guess whether $i_2\in Q_F^*$ or $i_2 \not\in Q_F^*$. If the guess is $i_2\in Q$, then add $i_2$ to $Q$. Together with the guesses from the previous step, this leads to a total of $\mathcal{O}(n^2)$ guesses.
           \item If the guesses are correct, then there exists an optimal solution $Q_F^*$ such that $i_1$ is the last item in $\bar{F}_{Q_F^*}$ and $i_2$ is the first item outside of $\bar{F}_{Q_F^*}$. That is, $i_1$ and $i_2$ are direct neighbors in the order $\bar{\prec}_{Q_F^*}$.
       \end{enumerate}
       \item Based on guesses of $i_1$ and $i_1$, and based on the current $Q$, which at this point is a subset of $\{i_1,i_2\}$, we partition $\items \setminus \{i_1,i_2\}$ into four sets $A,R,S_1$ and $S_2$. Provided that the guesses are correct, the sets $A$ and $R$ satisfy $A \subseteq Q_F^*$ and $R \cap Q_F^* = \emptyset$. Hence, we add $A$ to our solution $Q$ and omit the items in $R$. The items in $S_1$ and $S_2$ will be handled in consecutive steps.
       \begin{enumerate}
           \item Let $A := \{i \in \items \setminus Q \mid i_1 \bar{\prec}_{Q} i \bar{\prec}_{Q} i_2\}$ denote the set of items $i$ with an optimistic density $\odensity_i(Q)$ between $\odensity_{i_2}(Q)$ and $\odensity_{i_1}(Q)$. If the guesses of $i_1$ and $i_2$ is correct, then there exists an optimal solution $Q_F^*$ such that $i_1$ and $i_2$ are direct neighbors in the order $\bar{\prec}_{Q_F^*}$. Hence, we must have $A \subseteq Q_F^*$ as the items in $A \setminus Q_F^*$ would be between $i_1$ and $i_2$ in $\bar{\prec}_{Q_F^*}$.
           If any item $i \in A$ remains between $i_1$ and $i_2$ in the order $\bar{\prec}_{Q \cup A}$, then we immediately know that the current guess was incorrect. 
           \item Let $R := R_1 \cup R_2$ with $R_1 := \{i \in \items \setminus Q \mid i_2 \bar{\prec}_{Q} i\}$ and $R_2 := \{i \in \items \setminus Q \mid i \bar{\prec}_{Q} i_1 \land i_1 \bar{\prec}_{Q\cup \{i\}} i \bar{\prec}_{Q\cup \{i\}} i_2 \}$.     
           The items $i$ in $R_1$ are behind $i_2$ in the order $\bar{\prec}_{Q}$, which implies that they stay behind $i_2$ also in the order $\bar{\prec}_{Q \cup \{i\}}$. Provided that the guesses are correct, $R_1 \not\in \bar{F}(Q_F^*)$ and $R_1 \not\in \bar{F}(Q_F^* \setminus R_1)$, which implies $\bar{F}(Q_F^*) = \bar{F}(Q_F^* \setminus R_1)$. Hence, a correct guess implies $Q_F^* \cap R_1 = \emptyset$.
           The items in $R_2$ are before $i_1$ in the order $\bar{\prec}_Q$ and between $i_1$ and $i_2$ in the order $\bar{\prec}_{Q \cup \{i\}}$. That is, if an $i \in R_2$ is queried, it moves between $i_1$ and $i_2$  in the optimistic density order. If the guesses for $i_1$ and $i_2$ are correct, then there are no elements between $i_1$ and $i_2$ in the order $\bar{\prec}_{Q^*_F}$, which implies $R_2 \cap Q^*_F = \emptyset$. In particular, for correct guesses we have $R_2 \subseteq \bar{F}(Q^*_F)$.
           
           \item Let $S := \items \setminus (A \cup R \cup Q) = \{i \in \items \setminus (Q \cup R_2) \mid  i \bar{\prec}_Q i_1\}$. The items in $S$ come before $i_1$ in the order $\bar{\prec}_Q$. We further partition $S$ into the two subsets  $S_1 = \{i \in S \mid i_2 \bar{\prec}_{Q \cup \{i\}} i\}$ and $S_2 = S \setminus S_1$. That is, $S_1$ contains the set of items that come before $i_1$ in the optimistic density order $\bar{\prec}_Q$ but will move behind $i_1$ \emph{and} $i_2$ once they are queried. Since $S \cap R = \emptyset$, the definition of the set $R_2$ above implies that the items of $S_2$ will stay in front of $i_1$ in the optimistic density order, even if they are queried. 
       \end{enumerate}
       \item If our guesses are correct, then it remains to compute the subset of $S$ that should be added to $Q$. To this end, we guess $n_1 := |S_1 \cap Q_F^*|$ and $n_2 := |S_2 \cap Q_F^*|$. Together with the previous guesses, this leads to a total of $\mathcal{O}(n^4)$ guesses.  We proceed by computing the $n_1$ items of $S_1$ and the $n_2$ items of $S_2$ that should be added to $Q$:
       \begin{enumerate}
           \item 
           First, consider the set $S_1$. 
           Note that, for the current set $Q$, we already have that $i_1$ and $i_2$ are next to each other in the order $\bar{\prec}_Q$ as $Q$ contains the set $A$ from the previous step.
           In order to be consistent with our guesses for $i_1$ and $i_2$, we have to add a subset $P \subseteq S_1$ to $Q$ such that $i_1 \in \bar{F}(Q \cup P)$ and $i_2 \not\in \bar{F}(Q \cup P)$. Since $i_1$ and $i_2$ are already next to each other in $\bar{\prec}_Q$ and by definition of set $S_1$, this then implies that $i_1$ is last in $\bar{F}(Q \cup P)$ in the order $\bar{\prec}_{Q \cup P}$ and $i_2$ is first in $\items \setminus  \bar{F}(Q \cup P)$ in the order $\bar{\prec}_{Q \cup P}$, which is consistent with the guesses.

           To this end, let $K = \max\left\{\left(\sum_{j \bar{\prec}_Q i_1} w_j\right) - B + w_{i_1},0\right\}$ denote the minimum weight of items in $S_1$ that has to be queried for $i_1$ to enter the prefix. If $K = 0$, then $i_1$ is already part of $\bar{F}(Q)$. Similarly, let $H = \max\left\{\left(\sum_{j \bar{\prec}_Q i_2} w_j\right) - B + w_{i_2} -1,0\right\}$ denote the maximum amount of weight of items in $S_1$ that can be queried without $i_2$ entering the prefix.

           To be consistent with the guesses of $i_1$, $i_2$ and $n_1$, we need to select a subset $P \subseteq S_1$ with $K \le \sum_{i \in P} w_i \le H$ and $|P| = n_1$. Since the elements of $S_1 \setminus P$ will be part of the prefix $\bar{F}_{Q \cup P}$, we would like to minimize $\sum_{i  \in S_1 \setminus P} U_i = U_{S_1} - \sum_{i  \in S_1 \cap P} U_i$, which is equivalent to maximizing $\sum_{i  \in S_1 \cap P} U_i$. This leads to the following problem:

           \begin{equation}
               \tag{\ensuremath{\mathcal{P}_{S_1}}}
               \begin{array}{lll}
                   \max &\sum_{i \in S_1} x_i \cdot U_i\\
                   \text{s.t. }& \sum_{i \in S_1} x_i \cdot w_i \ge K & \\
                   & \sum_{i \in S_1} x_i \cdot w_i \le H & \\
                   & \sum_{i \in S_1} x_i = n_1 & \\
                   & x_i \in \{0,1\}& \forall i \in S_1
               \end{array}
           \end{equation}

           Our algorithm optimally solves~\eqref{density:sub:ILP} in pseudopolynomial time and adds the computed solution $P$ to $Q$.

           This can be done using the following dynamic program that slightly extends the textbook knapsack DP~\cite{IbarraK75}. Our goal is to compute the following DP-cells for all $i \in \{0,\ldots,|S_1|\}$, $b \in \{0,\ldots,D\}$ and $k \in \{0,\ldots, n_1\}$:
           $$
               T[i,b,k] := \max_{P \subseteq \{1,\ldots,i\} \colon |P| = k \land w(P) = b} \sum_{j \in P} U_j.
           $$
           If there is no packing $P \subseteq \{1,\ldots,i\}$ with $|P| = k$ and $w(P) = b$, then we want the DP-cell to store
           $$
              T[i,b,k] := - \infty.
           $$

           If we can correctly compute these DP-cells, then the optimal objective value for~\eqref{density:sub:ILP} is stored in one of the cells $T[|S_1|,b,k]$ with $K \le b \le H$ and $k = n_1 n_1$. If all these cells have value $-\infty$, then our guesses were certainly wrong.
           We proceed by describing how to compute the DP-cells.

           To this end, we use the following two base cases:
           \begin{itemize}
               \item $T[0,0,0] = 0$ since all packings $P$ with $w(P) = 0$ and $|P| = 0$ have $\sum_{j \in P} U_j = 0$.
               \item $T[0,b,k] = -\infty$ if $b > 0$ or $k > 0 $ as there exists no $P \subseteq \emptyset$ with $|P| > 0$ or $w(P) > 0$.
           \end{itemize}

           For $i \ge 1$, $b \in \{0,\ldots,D\}$ and $k \in \{0,\ldots, n_1\}$, we distinguish two cases:
           \begin{itemize}
               \item If $b-w_i < 0$ or $k-1 <0$, then the packing $P$ that maximizes 
               $$\max_{P \subseteq \{1,\ldots,i\} \colon |P| = k \land w(P) = b} \sum_{j \in P} U_j$$
               cannot contain $i$ and, thus,
               $$\max_{P \subseteq \{1,\ldots,i\} \colon |P| = k \land w(P) = b} \sum_{j \in P} U_j = \max_{P \subseteq \{1,\ldots,i-1\} \colon |P| = k \land w(P) = b} \sum_{j \in P} U_j.$$
               Provided that the cell $T[i-1,b,k]$ has been computed correctly, this implies $T[i,b,k] = T[i-1,b,k]$.
               \item Next, assume $b-w_i \ge 0$ and $k-1 \ge 0$. Then, we can rewrite 
               \begin{align*}
                   &\max_{P \subseteq \{1,\ldots,i\} \colon |P| = k \land w(P) = b} \sum_{j \in P} U_j =\\
                   & \max\left\{\max_{P \subseteq \{1,\ldots,i-1\} \colon |P| = k \land w(P) = b} \sum_{j \in P} U_j, U_i + \max_{P \subseteq \{1,\ldots,i-1\} \colon |P| = k-1 \land w(P) = b-w_i} \sum_{j \in P} U_j  \right\}
               \end{align*}
               by separating the packings $P$ that do not contain $i$ (first term in the maximum) and the packings $P$ that contain $i$ (second term of the maximum).
               Thus, we can compute the DP-cell as follows:
               $$T[i,b,k] = \max\{T[i-1,b-w_i,k-1] + U_i, T[i-1,b,k]\}.$$
           \end{itemize}
           Together, the two cases lead to the following recursive formula:
           $$T[i,b,k] = \begin{cases}
               \max\{T[i-1,b-w_i,k-1] + U_i, T[i-1,b,k]\} & \text{ if } b-w_i \ge 0 \land k-1 \ge 0\\
               T[i-1,b,k] & \text{ otherwise }
           \end{cases}$$

           The running time of this DP is in $\mathcal{O}(n^2 \cdot \sum_{i \in S_1} w_i)$. The optimal objective value for the instance of \eqref{density:sub:ILP} is contained in one of the cells $T[|S_1|,b,k]$ with $K \le b \le H$ and $k = n_1$. We can find this cell in time $\mathcal{O}(n^2 \cdot \sum_{i \in S_1} w_i)$ and compute the corresponding solution via backtracking.
           We omit the correctness proof for the DP, but remark that it can be shown using essentially the same proof as for the textbook knapsack DP.

           \item Next, we want to compute the elements of $S_2$ that should be added to $Q$. The current set $Q$ (including the elements added in the previous step 4a) was selected in such a way that $S_2 \subseteq \bar{F}(Q)$ and $S_2 \subseteq \bar{F}(Q \cup P)$ for every $P\subseteq S_2$. That is, the elements of $S_2$ are part of prefix $\bar{F}(Q)$ and will stay part of the prefix, even if they are queried. Thus, for every $P \subseteq S_2$, we have $U_{\bar{F}(Q)}(Q) - U_{\bar{F}(Q)}(Q \cup P) = \sum_{i \in P} (U_i-p_i)$. This implies that we should select the $n_2$ elements of $S_2$ with maximum $U_i-p_i$ and add them to $Q_2$. We can find these items in time $\mathcal{O}(|S_2| \log |S_2|)$ by sorting $S_2$.
       \end{enumerate}
       \item Among the sets $Q$ computed for the different guesses, return a set $Q$ of minimum cardinality subject to $U_{\bar{F}_Q}(Q) \le D$.
   \end{enumerate}

   \paragraph*{Running time.} The running time of the algorithm is in $\mathcal{O}(n^6 \cdot \sum_{i \in S_1} w_i)$ since there are $\mathcal{O}(n^4)$ guesses and the running time per guess is dominated by the running time $\mathcal{O}(n^2 \cdot \sum_{i \in S_1} w_i)$ of the DP in step~5a.

   \paragraph*{Correctness.} Since the algorithm tries all guesses for $i_1,i_2,n_1$ and $n_2$, there will be one iteration such that there exists an optimal solution $Q_F^*$ to the prefix problem such that $i_1$ is the last element in $\bar{F}(Q_F^*)$, $i_2$ is the first element in $\items \setminus \bar{F}(Q_F^*)$ in the order $\bar{\prec}_{Q_F^*}$ and $|S_1 \cap Q^*| = n_1$, $|S_1 \cap Q^*| = n_2$ for the sets $S_1$ and $S_2$ as computed for the guesses $i_1$ and $i_2$ in step~3. As argued above, $Q^*$ must satisfy $A \subseteq Q_F^*$ and $R \cap Q^* = \emptyset$ for the sets $S$ and $R$ as computed in step~3 for guesses $i_1$ and $i_2$. Furthermore, again as argued above,  we have $R_2 \subseteq \bar{F}(Q_F^*)$.

   By definition of the sets $S_1$ and $S_2$, we have $\bar{F}(Q_F^*) = S_2 \cup (S_1 \setminus Q_F^*) \cup R_2 \cup \{i_1\}$. Hence,
   $$
   U_{\bar{F}(Q_F^*)}(Q_F^*) = \sum_{j \in S_2} U_j - \left(\sum_{j \in S_2 \cap Q_F^*}  U_j - p_j\right) + \sum_{j \in S_1} U_j - \left(\sum_{j \in S_1 \cap Q_F^*}  U_j\right) + U_{R_2} + U_{i_1}(Q^*_F),
   $$
   with $|S_2 \cap Q_F^*| = n_2$, $|S_1 \cap Q_F^*| = n_1$ $K \le w(S_1 \cap Q_F^*) \le H$, where $H$ and $K$ are the parameter of~\eqref{density:sub:ILP} for the correct guesses.

   The set $Q$ computed by the algorithm for the correct guesses also satisfies $A \subseteq Q$ and $R \cap Q = \emptyset$. Since also $w(S_1 \cap Q^*) \le H$, we get $\bar{F}(Q) \subseteq S_2 \cup (S_1 \setminus Q) \cup R_2 \cup \{i_1\}$. Hence,
   $$
   U_{\bar{F}(Q)}(Q) \le  \sum_{j \in S_2} U_j - \left(\sum_{j \in S_2 \cap Q}  U_j - p_j\right) + \sum_{j \in S_1} U_j - \left(\sum_{j \in S_1 \cap Q}  U_j\right) + U_{R_2} + U_{i_1}(Q^*_F),
   $$
   where we use that $U_{i_1}(Q^*_F) = U_{i_1}(Q)$ holds as we correctly guessed whether $i_1 \in Q_F^*$.

   As $S_2$ maximizes $\sum_{j \in S_2 \cap Q}  U_j - p_j$ subject to $|S_2 \cap Q| \le n_2$ and $S_1$ maximizes $\sum_{j \in S_1 \cap Q}  U_j$ subject to $K \le w(S_1 \cap Q) \le H$ and $|S_1 \cap Q| = n_1$, we can conclude that $|Q| = |Q_F^*|$ and
   $$
   D \ge U_{\bar{F}(Q_F^*)}(Q_F^*) \ge  U_{\bar{F}(Q)}(Q),
   $$
   where the first inequality follows from $Q_F^*$ being an optimal solution for the prefix problem with parameter $D$. Thus, $Q$ is also an optimal solution.
\end{proof}

\thmPrefixNonPseudo*

\begin{proof}
   Our goal is to adjust the algorithm given in~\Cref{thm:prefix:pseudo} to achieve a polynomial running time at the cost of a worse guarantee. To this end, we replace step 4a and step 5 of the algorithm in~\Cref{thm:prefix:pseudo}. 

   We first argue how to replace step 4a. Observe that the only part of the algorithm given in~\Cref{thm:prefix:pseudo} with a pseudopolynomial running time is the subroutine for solving~\eqref{density:sub:ILP}. Recall that this subroutine is used to compute the subset $P \subseteq S_1$ that is added to the prefix problem solution $Q$.
   We replace this subroutine with the following polynomial time algorithm:
   \begin{enumerate}
       \item Consider the following relaxation of~\eqref{density:sub:ILP}, which drops the first constraint of~\eqref{density:sub:ILP} and removes the integrality constraint for the variables:
       \begin{equation}
           \tag{\ensuremath{\mathcal{P}'_{S_1}}}\label{density:sub:ILP:relax}
           \begin{array}{lll}
               \max &\sum_{i \in S_1} x_i \cdot U_i\\
               \text{s.t. } & \sum_{i \in S_1} x_i \cdot w_i \le H & \\
               & \sum_{i \in S_1} x_i = n_1 & \\
               & x_i \in \{0,1\}& \forall i \in S_1
           \end{array}
       \end{equation}
       \item Compute an optimal basic feasible solution $x^*$ of~\eqref{density:sub:ILP:relax}.
       \item Return $P = \{i \in S_1 \mid x^*_i = 1\}$ and add $P$ to $Q$.
   \end{enumerate}

   Finally, we adjust step 5 to return among all guesses the set $Q$ of minimum cardinality subject to $U_{\bar{F}(Q)}(Q) \le D  + 2 \cdot \max_{i \in \items} U_i$.

   We claim these two changes are sufficient to satisfy the theorem. First, note that the new step 4a has a polynomial running time as the running time is dominated by solving the LP. If we plug the subroutine into the algorithm of~\Cref{thm:prefix:pseudo}, this yields a polynomial running time. 

   It remains to show that $|Q'| \le |Q^*_F|$  and  $U_{\bar{F}(Q')}(Q') \le D + 2 \cdot \max_{i \in \items} U_i$ holds for the computed solution $Q$.
   The former holds as replacing the subroutine of step 4a with the approach above can only decrease the size of the solution since we omit the fractional variables.

   As shown in~\cite{CapraraKPP00}, an optimal basic feasible solution $x^*$ of~\eqref{density:sub:ILP:relax} has at most two fractional variables, i.e., at most two $i \in S_1$ have $1 > x^*_i > 0$. Let $P^*$ denote an optimal solution of~\eqref{density:sub:ILP}. Using that~\eqref{density:sub:ILP:relax} is a relaxation and that $x^*$ has at most two fractional values, we get
   \begin{equation}\label{eq:5}
       \sum_{i \in P} U_i  \ge \left(\sum_{i \in S_1} U_i \cdot x^*_i\right) - 2 \cdot \max_{i \in S_i} U_i \ge U_{P^*} - 2 \cdot \max_{i \in S_i} U_i.
   \end{equation}

   Finally, fix the solution $Q$ computed for the correct guesses. If this solution satisfies $|Q| \le |Q^*_F|$  and  $U_{\bar{F}(Q)}(Q) \le D + 2 \cdot \max_{i \in \items} U_i$, then the returned solution $Q'$ also satisfies $|Q'| \le |Q^*_F|$  and  $U_{\bar{F}(Q')(Q')} \le D + 2 \cdot \max_{i \in \items} U_i$.

   Note that the guesses are still with respect to an optimal solution $Q^*_F$ for the prefix problem with threshold $D$. With the same argumentation as above, we still get that
   $$
   U_{\bar{F}(Q_F^*)}(Q_F^*) = \sum_{j \in S_2} U_j - \left(\sum_{j \in S_2 \cap Q_F^*}  U_j - p_j\right) + \sum_{j \in S_1} U_j - \left(\sum_{j \in S_1 \cap Q_F^*}  U_j\right) + U_{R_2} + U_{i_1}(Q^*_F),
   $$
   with $|S_2 \cap Q_F^*| = n_2$, $|S_1 \cap Q_F^*| = n_1$ $K \le w(S_1 \cap Q_F^*) \le H$, where $H$ and $K$ are the parameter of~\eqref{density:sub:ILP} for the correct guesses, and 
   $$
   U_{\bar{F}(Q)}(Q) \le  \sum_{j \in S_2} U_j - \left(\sum_{j \in S_2 \cap Q}  U_j - p_j\right) + \sum_{j \in S_1} U_j - \left(\sum_{j \in S_1 \cap Q}  U_j\right) + U_{R_2} + U_{i_1}(Q^*_F),
   $$
   where we use that $U_{i_1}(Q^*_F) = U_{i_1}(Q)$ holds as we correctly guessed whether $i_1 \in Q_F^*$.

   As $S_2$ maximizes $\sum_{j \in S_2 \cap Q}  U_j - p_j$ subject to $|S_2 \cap Q| \le n_2$, we have  
   $$
   \sum_{j \in S_2} U_j - \left(\sum_{j \in S_2 \cap Q_F^*}  U_j - p_j\right) \ge \sum_{j \in S_2} U_j - \left(\sum_{j \in S_2 \cap Q}  U_j - p_j\right).
   $$
   Hence, 
   $$
   U_{\bar{F}(Q)}(Q) -  U_{\bar{F}(Q_F^*)}(Q_F^*) \le \left(\sum_{j \in S_1 \cap Q_F^*}  U_j\right) - \left(\sum_{j \in S_1 \cap Q}  U_j\right).
   $$
   By~\eqref{eq:5} and using that $ S_1 \cap Q_F^*$ is a feasible solution to~\eqref{density:sub:ILP} (as argued in the proof of~\Cref{thm:prefix:pseudo}), we get
   $$
   U_{\bar{F}(Q)}(Q) -  U_{\bar{F}(Q_F^*)}(Q_F^*) \le 2 \cdot \max_{i \in S_i} U_i,
   $$
   and, thus,
   $$
    U_{\bar{F}(Q)}(Q) \le  U_{\bar{F}(Q_F^*)}(Q_F^*) + 2 \cdot \max_{i \in S_i} U_i \le D + 2 \cdot \max_{i \in \items} U_i.
   $$
\end{proof}

\end{document}